\pdfoutput=1
\documentclass[11pt,a4paper]{article}

\usepackage[utf8]{inputenc}
\usepackage[T1]{fontenc}
\usepackage{lmodern}
\usepackage[margin=1in]{geometry}
\usepackage{amsmath,amssymb,amsthm,bm}
\usepackage{graphicx}
\usepackage{booktabs}
\usepackage{longtable}
\usepackage{algorithm}
\usepackage{algpseudocode}
\usepackage{xcolor}
\usepackage{tabularx}
\usepackage{ragged2e}
\usepackage{authblk}
\usepackage{tikz}
\usetikzlibrary{arrows.meta,positioning,fit,calc}
\usepackage[section]{placeins}
\usepackage[numbers,sort&compress]{natbib}
\usepackage{hyperref}
\hypersetup{colorlinks=true,urlcolor=blue,linkcolor=blue,citecolor=blue}
\newcolumntype{Y}{>{\RaggedRight\arraybackslash}X}

\newtheorem{proposition}{Proposition}
\theoremstyle{remark}\newtheorem{remark}{Remark}

\newcommand{\bx}{\bm{x}}
\newcommand{\bu}{\bm{u}}

\newcommand{\bff}{\bm{f}}
\newcommand{\be}{\bm{e}}
\newcommand{\bM}{\bm{M}}

\newcommand{\bK}{\bm{K}}
\newcommand{\bF}{\bm{F}}
\newcommand{\bP}{\bm{P}}
\newcommand{\bR}{\bm{R}}
\newcommand{\bD}{\bm{D}}
\newcommand{\bI}{\bm{I}}
\newcommand{\eps}{\varepsilon}

\title{\textbf{A Matrix-Free Galerkin Multigrid Solver and Failure-Mode Screen for\\
Single-GPU 3D SIMP Linear Systems}}

\author[1]{Shaoliang Yang}
\author[1]{Jun Wang\thanks{Corresponding author. E-mail: jwang22@scu.edu}}
\author[1]{Yunsheng Wang}
\affil[1]{Department of Mechanical Engineering, Santa Clara University, Santa Clara, CA 95053, USA}
\date{}

\begin{document}
\maketitle

\begin{abstract}
Large 3D SIMP studies require repeated elasticity solves for density-dependent
operators whose finest matrices are expensive to assemble and whose conditioning
degrades under high contrast.
We study this linear-solver layer rather than claiming end-to-end optimization
acceleration.
The solver builds a matrix-free Galerkin geometric multigrid (GMG) hierarchy
around a fused fine operator: the finest level remains matrix-free, the first
coarse level is assembled by local Galerkin aggregation, and deeper levels use
sparse Galerkin products.
The practical default is FP32-GMG; BF16 is evaluated as a guarded
mixed-precision variant and diagnostic stress test, not as the main speed
mechanism.
In a 27-case heterogeneous cantilever sweep, pass rates under a 200-iteration
budget are 7/9, 4/9, and 1/9 at 64k, 216k, and 512k elements; converged-only
mean iteration counts are about 112, 134, and 146.
On uniform $\rho=0.5$, $p=3$ solves, FP32-GMG gives $1.62\times$,
$1.75\times$, and $3.12\times$ wall-time ratios relative to the capped
flat Jacobi-PCG baseline at the same sizes; that non-converged baseline
reaches the 200-iteration cap in all timed trials.
BF16-GMG is not faster than FP32-GMG.
In 18 fixed-seed heterogeneous BF16 validation cases, 7/18 converge, matching
the FP64 count, and 11 cases that pass the spectral screen still fail the
500-iteration cap; the screen is therefore diagnostic rather than a convergence
certificate.
The largest reported solve is a 1M-element uniform-modulus system solved in
$1.50\pm0.58$\,s with an 8.66\,GiB hierarchy-allocation delta during setup,
not a peak-memory trace; this point is reported as uniform scaling, not
heterogeneous robustness evidence.
The contribution is therefore a bounded single-GPU solver result built on an
inherited Level~0 matrix-free operator: a Galerkin GMG hierarchy, direct BF16
guard evidence, and an explicit failure-mode screen for structured 3D SIMP
linear systems.
\end{abstract}


\section{Introduction}\label{sec:intro}

Topology optimization has evolved from homogenization-based formulations
\citep{bendsoe1988homogenization} to density-based SIMP workflows
\citep{bendsoe1999material,andreassen201188} and, more broadly,
to projection-, filter-, and level-set-based methods that regularise the
otherwise ill-posed material-distribution problem
\citep{bourdin2001filters,guest2004minimum,wang2003levelset,allaire2004levelset,sigmund2013review}.
For minimum-compliance structural design, the dominant computational cost is
not the optimiser but the repeated solution of large linear elasticity systems
as the density field $\bm{\rho}$ evolves.  That cost becomes decisive in
three-dimensional runs, where practical studies increasingly target
$10^5$--$10^7$ elements, multiple benchmark families, and aggressive SIMP
continuation schedules \citep{aage2013parallel,aage2015petsc,aage2017giga,yago2022comparative}.

At that scale, the central tension is architectural.
Classical parallel frameworks based on assembled sparse matrices and AMG
preconditioners remain highly effective on clusters
\citep{aage2015petsc,aage2017giga,herrero2023adaptive}, but their memory model
does not map cleanly onto a single consumer GPU.
Recent GPU-oriented topology-optimization systems therefore move toward
matrix-free operators, structured grids, and hardware-conscious kernels
\citep{traff2023simple,padhi2023gpu,zhao2024gpu}.
In that setting, the stiffness action is naturally expressed as a fused
gather--GEMM--scatter kernel that never materialises the global stiffness
matrix and instead streams element data through tensor-core-friendly dense
micro-kernels \citep{yang2026fused}.
In the present work that fine operator is treated as fixed level-0
infrastructure; the question is not how to redesign the element operator, but
how to construct a solver hierarchy around it.

That architectural choice solves the memory problem but exposes a numerical one.
For the flat Jacobi-preconditioned path studied here and in the companion
operator paper, BF16 is unreliable unless the effective spectrum is compressed
far below what a flat diagonal preconditioner delivers.
The mixed-precision literature makes the issue explicit: low-precision inner
solves are reliable only when $\varepsilon\kappa$ is controlled
\citep{carson2017new,higham2022mixed,haidar2020tensor}.
For the SIMP systems studied here,
$\kappa(\bD^{-1}\bK)\sim10^3$--$10^5$, so BF16
($\eps_{\mathrm{BF16}}\approx3.9\times10^{-3}$) fails that requirement by a
wide margin.
The implication is structural rather than incidental:
if BF16 is to be used reliably in single-GPU SIMP state solves, the solver
hierarchy must reduce the effective spectrum first.

The resulting problem is therefore specific and self-contained.
Rather than redesigning the fine matrix-free operator, this paper asks whether
a Galerkin hierarchy can make structured 3D SIMP state solves practical on one
GPU in benign-to-moderate regimes, what its memory and setup costs are, and
where the same hierarchy fails on high-contrast density fields.  BF16 is
studied as a guarded mixed-precision path, but the solver contribution is the
matrix-free Galerkin hierarchy and the resulting failure-mode screen.

\paragraph{This paper.}
We replace the flat Jacobi preconditioner with a \emph{matrix-free Galerkin
geometric multigrid} hierarchy and measure how far it improves frozen SIMP
linear solves before high-contrast density fields defeat the present coarse
space and smoother balance.
The fused fine-level operator is treated here as a fixed level-0 building
block rather than as a contribution of this paper.
The evidence is solver-level: frozen-coefficient linear solves,
spectral-proxy sweeps, ablations, and scaling.  The auxiliary fixed-penalty
30-step schedule is retained only as an auxiliary diagnostic and is not part of
the main optimization evidence chain.
The concrete contributions are:

\begin{enumerate}
\item \textbf{Matrix-free Galerkin hierarchy.}
  The finest level remains fully matrix-free; Level~1 is assembled
  element-by-element as $\bm{K}_1 = \bP^\top\bm{K}^{(0)}\bP$ without ever
  forming $\bm{K}^{(0)}$; coarser levels use exact sparse triple products.
  This preserves Galerkin consistency while eliminating the assembled
  fine-stiffness bottleneck at the finest level.

\item \textbf{Mixed-precision schedule and spectral proxy.}
  BF16 is tested at the finest level, FP32 at Level~1, and FP64 on deeper
  levels.  We give an idealised conditional estimate that motivates
  $\kappa_{\mathrm{eff}}$ as a screening quantity, then treat
  $\eps_{\mathrm{BF16}}\kappa_{\mathrm{eff}}<1$ as an empirical spectral proxy
  rather than as a proof of BF16 stability for the implemented hierarchy.
  Direct BF16 solves on the heterogeneous screening states are used to report
  where that diagnostic succeeds and where it misclassifies convergence.

\item \textbf{Chebyshev-capable hierarchy and smoother ablation.}
  Integrates a degree-$\nu$ Chebyshev polynomial smoother into the hierarchy,
  keeps damped Jacobi as a matched comparison path, and reports when the
  current Chebyshev default should be read as a provisional
  robustness-oriented engineering choice rather than as the fastest option on
  mild states.

\item \textbf{Empirical convergence and failure-mode screen.}
  The reported experiments document where the GMG hierarchy converges on the
  tested suite and where heterogeneous 512k or near-singular stress cases mark
  the current failure boundary.
\end{enumerate}

\paragraph{Scope vs.\ companion paper.}
The present manuscript builds on the companion operator/kernel study of
\citet{yang2026fused}
(\href{https://arxiv.org/abs/2604.18020}{arXiv:2604.18020}), which introduced
the fused gather--GEMM--scatter Level~0 matrix-free SIMP operator used here.
That companion paper owns the Level~0 operator and fused-kernel design,
including the operator-level roofline characterisation and the flat-Jacobi
baseline context.  The present manuscript does not claim that operator as a
contribution.  It asks what solver can be built around the fixed Level~0
action: the matrix-free Galerkin GMG hierarchy, the Jacobi/Chebyshev smoother
policy and ablations, the precision-descent schedule, and the
$\eps_{\mathrm{BF16}}\!\cdot\!\kappa_{\mathrm{eff}}$ admissibility analysis in
Appendix~\ref{app:bound} are contributions of this paper.  The negative
flat-Jacobi BF16 result from the companion paper motivates this hierarchy: the
diagonally scaled condition number $\kappa(\bD^{-1}\bK)$ was too large for a
flat preconditioner to make BF16 reliable at the tested SIMP scales.  The
present paper is self-contained at the solver level: it specifies the
hierarchy, smoother policy, precision schedule, stopping rules, benchmark
states, and result provenance needed to reproduce the reported GMG experiments.
All solver variants compared here use the same Level~0 action, so the
quantitative comparisons isolate the multigrid and precision-policy choices
rather than re-evaluating the fused-kernel design.  The primary methodological
and solver-level claims are therefore distinct, although the present paper
reprises a limited operator-level proxy measurement only to interpret solver
behavior.

\paragraph{Relation to prior work.}
Multigrid for topology optimization is not new
\citep{amir2014multigrid,wu2016system,wang2025matrixfree,herrero2023adaptive}, nor is the use of
mixed precision in Krylov or multigrid solvers
\citep{carson2017new,higham2022mixed,tsai2023mixed,tsai2023three}.
Relative to the cited single-GPU 3D SIMP, matrix-free 3D TO, and
mixed-precision multigrid literature, the present paper studies a more
specific conjunction: a matrix-free fine operator \citep{yang2026fused}, a
structured Galerkin hierarchy that does not assemble the finest matrix, and a
precision schedule that tests BF16 where tensor cores can be used while
checking true residuals and a spectral proxy.
AMG-based solvers remain the closest external baseline family: AMGx is a GPU
representative and PyAMG provides the assembled CPU reference used in
our assembled-baseline comparison \citep{naumov2015amgx,bell2023pyamg}.

\paragraph{Paper organisation.}
Section~\ref{sec:related} reviews multigrid for TO, mixed-precision Krylov,
and GPU AMG.
Section~\ref{sec:method} describes the V-cycle design, Galerkin coarse
operators, and the precision-descent schedule.
Section~\ref{sec:results} then combines implementation details, validation
checks, and the solver-focused experimental study so that solver design and
measured behavior are read in one continuous arc.
Section~\ref{sec:gallery} provides a compact qualitative structure gallery for
mechanics readers; it is visual context rather than quantitative solver
evidence.
Section~\ref{sec:discussion} discusses failure modes and extensions.
Section~\ref{sec:conclusion} restates contributions and open questions.
Appendix~\ref{app:repro} documents reproducibility;
Appendix~\ref{app:artifact} gives the reproduction workflow and result provenance;
Appendix~\ref{app:bound} derives the $\kappa_{\mathrm{eff}}$ bound.

\section{Related Work}\label{sec:related}

\subsection{Topology Optimization Formulations and Regularization}

Modern structural topology optimization still rests on the conceptual split
introduced by the early homogenization and density formulations:
one chooses a material parameterisation, regularises it sufficiently to obtain
meaningful limits, and solves a large sequence of state equations
\citep{bendsoe1988homogenization,bendsoe1999material}.
For density methods, regularization is not peripheral.
Convolution filters \citep{bourdin2001filters}, projection-based minimum-length
controls \citep{guest2004minimum}, and robust density regularization are what turn a numerically unstable
binary design problem into a reproducible computational workflow.
Alternative representations, especially level-set methods
\citep{wang2003levelset,allaire2004levelset} and comparative
reviews across method families \citep{sigmund2013review,yago2022comparative},
make the same point from another angle: the optimization algorithm can change,
but the linear-elasticity solve remains the central cost driver once the mesh
becomes large.

\subsection{Large-Scale and Parallel Topology Optimization}

The move from educational 2D codes to research-grade 3D systems has largely
been driven by parallel implementation.
\citet{aage2013parallel} established the distributed-memory optimization
framework: MMA, structured brick elements, and parallel sparse linear algebra
were organised as a scalable CPU workflow rather than as a GPU kernel problem.
\citet{aage2015petsc} made that workflow concrete in PETSc, with filtering,
structured grids, and an openly reproducible parallel implementation.
\citet{aage2017giga} then showed how far the assembled-cluster paradigm can be
pushed, reaching giga-voxel scale on a supercomputer.  Read together, these
papers fix an important baseline for the present work.
They show that large 3D SIMP is feasible when one accepts a distributed sparse
matrix representation and cluster memory budget; they do not solve the
workstation problem in which the assembled finest matrix is itself the dominant
bottleneck.

That distinction matters because the single-GPU regime is not a smaller copy of
the cluster regime.  \citet{ferrari2020top99neo} and
\citet{yago2022comparative} emphasize implementational economy and benchmark
comparability, but the underlying linear systems remain small enough that
assembly is not the limiting architectural decision.  By contrast, once the
target platform is a single consumer GPU, the central question becomes how much
of the sparse hierarchy can be replaced by structured transfer operators,
matrix-free finest-level actions, and bandwidth-conscious kernels without
giving up robustness.

\subsection{Multigrid for Topology Optimization}

Multigrid is a natural response to the linear-solver bottleneck, but its role
in topology optimization is more delicate than in textbook Poisson problems.
\citet{amir2014multigrid} is the closest classical antecedent: it showed that
multigrid-preconditioned Krylov solvers can make structured-hex SIMP efficient,
but the implementation target was a conventional sparse-matrix CPU setting.
The present paper inherits Amir et al.'s insistence that the solver hierarchy
must be designed jointly with the optimization problem;
it departs by refusing to assemble the finest matrix at all.

\citet{wu2016system} pushed topology optimization toward high-resolution
single-device execution and is therefore closer in spirit to the current
hardware target.  Its emphasis, however, is high-resolution GPU execution for
the optimization pipeline as a whole rather than a mixed-precision GMG
hierarchy with an explicit stability condition.  \citet{traff2023simple}
occupies a similarly important position: it demonstrates how far a carefully
tuned structured GPU implementation can go while remaining close to standard
TO ingredients.  We inherit the same design preference for structured grids and
aggressive kernel tuning, but we depart from Tr\"aff et al. by testing whether
the linear solver itself can make BF16 fine smoothing viable on selected
states, with true-residual checks and an explicit failure-mode screen.

An earlier CUDA antecedent is \citet{gavranovic2019gpgpu}, who implemented
matrix-free geometric multigrid for topology optimization on regular
structured hexahedral meshes.  We treat that work as an important
GPU/matrix-free GMG predecessor; the distinction here is the Galerkin hierarchy
around the inherited fused Level~0 SIMP operator together with BF16/FP32/FP64
level scheduling, direct BF16 true-residual validation, and failure-mode
screening.

\citet{wang2025matrixfree} is the closest recent matrix-free 3D TO neighbour
on the solver-design side.  Their MATLAB framework combines an element-based
matrix-free action with GMG and pushes structured-hex 3D SIMP to very large
problem sizes on a standard workstation.  The present paper departs in target
platform and numerical scope: it is a single-consumer-GPU study rather than a
MATLAB/workstation study, it keeps the finest operator resident in the GPU-side
matrix-free path throughout the hierarchy experiments, and it makes the
mixed-precision conditioning question---especially BF16 fine smoothing under a
Galerkin hierarchy---a first-class claim rather than a secondary implementation
detail.

\citet{padhi2023gpu} is another important GPU-side multigrid neighbour.  Their hybrid
GPU and homogenization-based multigrid approach recognises the same pressure
point as the present work: three-dimensional TO becomes solver-limited long
before the optimiser becomes the bottleneck.  The difference is architectural.
Padhi et al.\ operate with a homogenization-driven multigrid construction,
whereas our hierarchy is explicitly Galerkin and is built to preserve the
fine-operator correction problem seen by the outer solver.  That distinction is
not cosmetic here because the present paper uses the coarse hierarchy to bound
the effective condition number seen by a BF16 fine smoother.

\citet{herrero2023adaptive} tackle the multigrid problem from a different end:
adaptive non-conforming meshes, distributed memory, and parallel GMG for
Krylov-preconditioned topology optimization.  Their work shows that geometric
multigrid remains viable even when the mesh hierarchy is irregular and
distributed.  We inherit the message that geometry-aware transfer operators are
worth preserving, but we depart in three ways: the present hierarchy is
single-GPU rather than distributed-memory, matrix-free at the finest level
rather than assembled, and designed around mixed-precision conditioning diagnostics rather
than parallel AMR.

\citet{lazarov2011filters} showed that Helmholtz filtering introduces a
PDE-defined length scale into the filtered design field.  In the present work,
where that field defines the fine SIMP stiffness operator, we therefore require
any coarse-grid correction to remain consistent with the filtered fine
operator.  This is one
reason the present paper takes the Galerkin side of the classical
Galerkin-versus-re-discretization tradeoff
\citep{trottenberg2000multigrid,briggs2000multigrid,bohm2025galerkin}.  A
hierarchy whose coarse operators drift away from the filtered fine operator may
still accelerate a solve, but it does not help the specific BF16 stability
problem addressed here.

\subsection{Matrix-Free GPU Finite Element Solvers}

The numerical linear-algebra side of this paper is also tied to a broader HPC
trend toward matrix-free finite element methods on GPUs.
Libraries such as MFEM emphasize matrix-free operators, high-order kernels,
and backend portability precisely because assembled matrices often waste both
memory and bandwidth on accelerator hardware \citep{anderson2021mfem}.
Related work on low-order refined preconditioners, matrix-free GPU flow
solvers, and vectorisation studies
\citep{franco2020lor,vargas2022matrixfree,sun2020vectorization} reaches the
same conclusion from different applications: the viable accelerator strategy is
usually to spend arithmetic to save memory traffic, then design the
preconditioner around that choice.
\citet{davydov2020matrixfree} make the same argument in nonlinear solid
mechanics: once the operator is matrix-free, multigrid has to be rebuilt around
that representation rather than assumed to come for free from an assembled CSR
matrix.  That paper is not about topology optimization, but it is directly
relevant methodologically because it demonstrates that matrix-free finite
elements and geometric multigrid can coexist in solid-mechanics settings when
the level operators are constructed with care.

Within topology optimization, \citet{traff2023simple} demonstrated how far a
carefully tuned structured GPU code with multigrid-preconditioned solves can
go without making mixed-precision GMG the central contribution, while recent
single-GPU matrix-free baselines already show that the fine operator itself can
be reformulated as a fused element-level kernel.
The present paper studies a complementary point in that design space: a
single-GPU Galerkin hierarchy and mixed-precision diagnostic wrapped around a
fixed matrix-free fine operator.

\subsection{Mixed Precision, Tensor Cores, and GPU AMG}

Reduced precision becomes attractive on GPUs because tensor cores move dense
small-matrix arithmetic into a much higher-throughput regime
\citep{nvidia2023ada,markidis2018nvidia}.
The obstacle is not hardware support but numerical stability.
Iterative refinement and mixed-precision Krylov theory
\citep{carson2017new,higham2022mixed,haidar2020tensor} state the condition in
its simplest form: low-precision inner work is safe only when the effective
conditioning seen by that work is sufficiently small.
The multigrid literature makes the same point from older and newer directions.
\citet{goddeke2007mixedfem} and \citet{goddeke2011cyclic} already showed that
mixed precision can be profitably embedded in geometric multigrid, but only
when the accuracy-critical parts of the cycle remain protected.  The same logic
reappears in more recent AMG work: \citet{tsai2023mixed} and
\citet{tsai2023three} allow different precisions on different AMG levels, and
\citet{kashi2026mixed} surveys the broader design space, but the consistent
message is that lower precision is viable only after one decides which parts of
the multigrid cycle are spectrum-sensitive and which are bandwidth-limited.

The tensor-core linear-algebra literature is relevant here for a more local
reason.  \citet{abdelfattah2019batched} and \citet{abdelfattah2020batched}
study precisely the regime of small dense blocks in which the present
matrix-free element kernel operates.  \citet{lopez2023mixedlu} and
\citet{haidar2020tensor} then show how mixed-precision tensor-core kernels must
be organised so that low-precision multiplies are buffered by higher-precision
accumulation and correction.  \citet{buttari2008sparse} and
\citet{buttari2007dense} provide the older iterative-refinement viewpoint:
reduced precision is useful when it is embedded inside a higher-precision outer
iteration with an explicitly controlled conditioning requirement.

Domain-specific FEM studies reinforce the same message.  The question is therefore not whether mixed precision can help a
linear solve in principle, but where it may be inserted in a topology
optimization solver without breaking the solve.  On the sparse-solver side,
AMGx provides a strong GPU baseline for assembled CSR systems
\citep{naumov2015amgx}, and PyAMG remains a standard CPU reference
\citep{bell2023pyamg}.  Both are valuable comparators, but neither solves the
specific problem targeted here: a matrix-free SIMP solver designed to use BF16
at the fine smoothing stage without materialising the fine stiffness matrix.

\begin{table}[!tbp]
\centering
\footnotesize
\caption{Nearest prior-art comparison.  Entries use directly stated information
from the cited papers where available; ``N/R'' means that the metric is not
reported explicitly in the source paper.}
\label{tab:prior_art}
\setlength{\tabcolsep}{3pt}
\begin{tabularx}{\linewidth}{>{\RaggedRight\arraybackslash}p{0.12\linewidth}>{\RaggedRight\arraybackslash}p{0.12\linewidth}>{\RaggedRight\arraybackslash}p{0.15\linewidth}Y >{\RaggedRight\arraybackslash}p{0.10\linewidth}>{\RaggedRight\arraybackslash}p{0.17\linewidth}}
\toprule
Work & Hardware & Scale & Solver / hierarchy & Prec. & Reporting \\
\midrule
\citet{aage2013parallel} & Distributed CPU & Large 3D TO; ceiling varies by benchmark & Parallel sparse FEM + MMA framework & FP64 & N/R \\
\citet{aage2015petsc} & Distributed CPU/PETSc & Structured 3D TO & Assembled sparse PETSc stack & FP64 & N/R \\
\citet{aage2017giga} & HPC cluster & Giga-voxel & Cluster-scale assembled/distributed workflow & FP64 & Design scale; solver timing N/R \\
\citet{wu2016system} & Single GPU & High-res. TO & GPU-accelerated structured solve pipeline & FP32 & Runtime scaling; no mixed precision \\
\citet{gavranovic2019gpgpu} & Single NVIDIA GPU & Structured hexahedral TO & CUDA matrix-free geometric multigrid & N/R & Runtime reduction claimed; no BF16 tensor cores \\
\citet{wang2025matrixfree} & 64\,GB PC; \mbox{MATLAB} & Large 3D TO & Matrix-free fine action + GMG & MATLAB path & Scale and wall time; no GPU/BF16 \\
\citet{traff2023simple} & Single GPU & Large 2D/3D TO & Structured GPU TO with V-cycle MG & FP32 auxiliaries / FP64 MG ops & Wall time reported; no BF16 tensor cores \\
\citet{padhi2023gpu} & CPU+GPU hybrid & Large 3D TO & Homogenization multigrid & FP64 / FP32 & Wall time; no BF16 tensor cores \\
\citet{herrero2023adaptive} & Distributed CPU & Large adaptive 3D TO & Parallel GMG on adaptive meshes & FP64 & Scaling \\
\textbf{This paper} & RTX 4090 & 3.09M free DOFs / 1M elements & Matrix-free fine action + Galerkin GMG & BF16, FP32, FP64 by level & Iterations, timing, VRAM, BF16 diagnostic \\
\bottomrule
\end{tabularx}
\end{table}

\paragraph{Positioning relative to the closest work.}
Relative to \citet{amir2014multigrid} and \citet{herrero2023adaptive}, the
present paper narrows the hardware target from distributed sparse-matrix
multigrid to a single consumer GPU and therefore makes operator representation
the central design decision.  Relative to \citet{wang2025matrixfree},
\citet{traff2023simple}, \citet{wu2016system}, and \citet{padhi2023gpu}, it
narrows the numerical claim: the point is not merely that a matrix-free or GPU
implementation is fast, but that a matrix-free fine operator can be embedded in
a mixed-precision Galerkin hierarchy on a single consumer GPU and checked
against an effective-conditioning diagnostic for guarded BF16 fine smoothing.
Relative to prior matrix-free Jacobi-PCG baselines, the distinguishing claim
here is not merely matrix-free execution but matrix-free
preconditioning: the fine BF16 kernel is embedded in a Galerkin hierarchy
designed to test whether the low-precision stage can be guarded by a bounded
effective-spectrum diagnostic.

\section{Methodology}\label{sec:method}

\begin{table}[!tbp]
\centering
\small
\caption{Principal symbols used in Section~\ref{sec:method}.}
\label{tab:nomenclature}
\begin{tabular}{p{0.18\linewidth}p{0.74\linewidth}}
\toprule
Symbol & Meaning\\
\midrule
$\bm{\rho},\tilde{\bm{\rho}},\hat{\bm{\rho}}$ & Design, filtered, and physical density fields\\
$\bF$ & Density-filter matrix\\
$E_e$ & Element modulus induced by SIMP interpolation\\
$E_{\min}$ & Modulus floor in the SIMP interpolation law\\
$\rho_{\mathrm{floor,test}}$ & Synthetic density-floor label used to construct standalone stress fields; it is distinct from $E_{\min}$\\
$V_f,p$ & Prescribed volume fraction and SIMP penalization used to construct a benchmark state\\
$\bK^\ell$ & Level-$\ell$ stiffness operator on the free-DOF subspace\\
$\bP^\ell,\bR^\ell$ & Prolongation and restriction between levels $\ell$ and $\ell+1$, with $\bP^\ell\!:\ell\!+\!1\rightarrow\ell$\\
$\bD$ & Diagonal of the active level operator used by the smoother\\
$\mathcal{M}$ & Frozen left-preconditioning operator induced by one multigrid V-cycle on the free-DOF system\\
$\kappa_{\mathrm{eff}}$ & Effective condition number of the applied left-preconditioned operator $\mathcal{M}\bK$\\
$L$ & Coarsest hierarchy level\\
$\nu$ & Polynomial smoother degree\\
\bottomrule
\end{tabular}
\end{table}

\subsection{Problem Setup}

For the optimization problem we distinguish the design density
$\bm\rho$, the filtered density $\tilde{\bm\rho}=\bF\bm\rho$, and the physical
density $\hat{\bm\rho}=H_{\beta,\eta}(\tilde{\bm\rho})$ after the optional
Heaviside projection.  The constitutive law is evaluated from the physical
field,
\begin{equation}
  E_e(\hat\rho_e)=E_{\min}+(E_0-E_{\min})\hat\rho_e^p,
  \qquad
  \hat{\bm\rho}=H_{\beta,\eta}(\bF\bm\rho),
  \label{eq:simp_filter}
\end{equation}
with $E_{\min}=10^{-9}$ in the reported runs.  Each linear solve therefore sees
a frozen modulus field $\{E_e\}$ associated with one SIMP state.  The filter is
applied once per SIMP update; it is not re-applied inside the multigrid cycle.
In the standalone linear-solver studies, we prescribe $\{E_e\}$ directly so
that solver effects are isolated from optimiser dynamics.
We reserve $E_{\min}$ for the modulus floor in \eqref{eq:simp_filter}.
By contrast, $\rho_{\mathrm{floor,test}}$ denotes a synthetic density-construction
label in the standalone stress fields.  When such a field is mapped through
\eqref{eq:simp_filter}, $E_{\min}$ still sets the final stiffness floor; this
is why the very-low $\rho_{\mathrm{floor,test}}$ labels in the low-floor basin screen
collapse to essentially the same effective modulus floor.

With this frozen-coefficient view, the state equation is
\begin{equation}
  \bK(\hat{\bm\rho})\bu = \bff,
  \qquad
  \bK(\hat{\bm\rho}) = \sum_e E_e(\hat\rho_e)\bm{K}_e,
  \label{eq:state_eq}
\end{equation}
where $\bm{K}_e$ is the $24\!\times\!24$ unit element stiffness matrix for a
trilinear hexahedral element with $\nu=0.3$.
The free-DOF system is $\bm{K}_{\mathrm{ff}}\bu_f = \bff_f$.

\begin{figure}[!tbp]
\centering
\includegraphics[width=\linewidth]{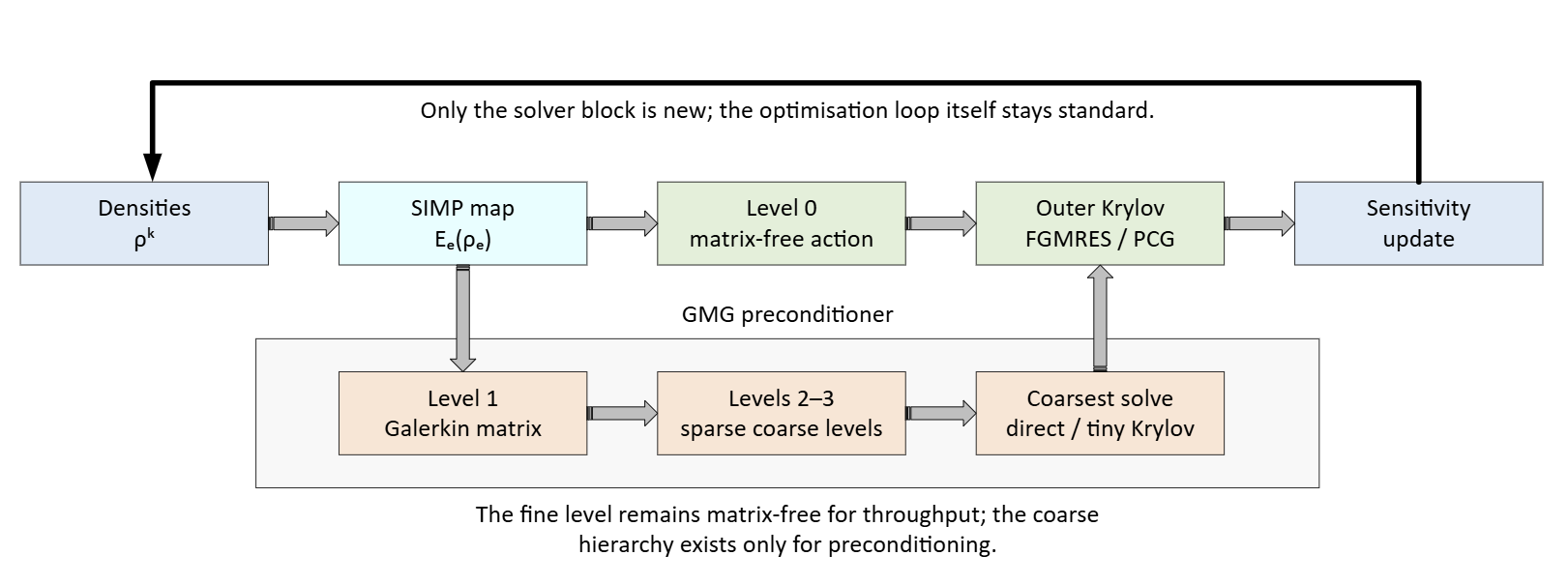}
\caption{Solver placement inside one SIMP iteration.
  The density update remains the standard optimization loop; the change in this
  paper is the linear-solver stack that sits between the current density field
  and the sensitivity update.
  Level~0 stays fully matrix-free and applies the fixed fused fine operator
  of \citet{yang2026fused}, while Level~1 and below form a Galerkin hierarchy
  used only for preconditioning.
  The figure locates the contribution architecturally before the performance
  and robustness plots appear.}
\label{fig:method_overview}
\end{figure}

\subsection{Matrix-Free V-Cycle}

The V-cycle at level $\ell$ applies:
\begin{equation}
  \mathcal{V}^\ell(\bm{r}) =
  \begin{cases}
    \mathcal{S}^\ell_{\mathrm{post}}\!\left(
      \mathcal{S}^\ell_{\mathrm{pre}}(\mathbf{0}, \bm{r})
      + \bP^\ell \mathcal{V}^{\ell+1}(\bR^\ell(\bm{r} - \bK^\ell \mathcal{S}^\ell_{\mathrm{pre}}(\mathbf{0},\bm{r})))
    \right)
    & \ell < L,\\
    (\bK^L)^{-1}\bm{r} & \ell = L,
  \end{cases}
  \label{eq:vcycle}
\end{equation}
where $\bP^\ell$ and $\bR^\ell = (\bP^\ell)^\top$ are the prolongation and
restriction operators, $\mathcal{S}^\ell$ is the smoother, and $\bK^L$ is
the coarsest-level operator solved by dense Cholesky (if $n_L \le 5000$)
or PCG.

\subsection{Transfer Operators and Boundary Masking}

Coarsening is standard $2{:}1$ in each Cartesian direction on the structured
hexahedral grid.  The scalar prolongation operator $P_s^\ell$ interpolates from
coarse nodes $(i_c,j_c,k_c)$ to fine nodes $(i_f,j_f,k_f)$ with tensor-product
trilinear weights,
\begin{equation}
  (P_s^\ell)_{i_f j_f k_f,\; i_c j_c k_c}
  =
  w_x(i_f,i_c)\,w_y(j_f,j_c)\,w_z(k_f,k_c),
  \label{eq:trilinear_weights}
\end{equation}
where each one-dimensional weight is either $1$, $0$, or $\tfrac12$: even fine
indices inject to one parent coarse node, and odd fine indices split equally
between the two neighbouring coarse nodes.
The vector-valued transfer operator is
\begin{equation}
  \bP^\ell = P_s^\ell \otimes \bI_3,
  \qquad
  \bR^\ell = (\bP^\ell)^\top.
  \label{eq:vector_prolongation}
\end{equation}

Dirichlet constraints are inherited by \emph{injection}, not by prolongation
support.  A coarse node $(i_c,j_c,k_c)$ is mapped to the fine node
$(2i_c,2j_c,2k_c)$, and a coarse DOF is declared free if and only if that
injected fine DOF is free.  This avoids the boundary-pathology in which a
fixed-face coarse node would otherwise receive fractional prolongation support
from adjacent free interior nodes and be misclassified as free.

\paragraph{Fine-level operator ($\ell=0$).}
$\bK^0$ is never assembled.  Matrix--vector products are performed by the
fused gather--GEMM--scatter kernel:
\begin{equation}
  (\bK^0\bu)_{\mathrm{free}} = \mathcal{F}_{\mathrm{mf}}(\bu;\{E_e\}).
  \label{eq:mf_matvec}
\end{equation}
Here $\mathcal{F}_{\mathrm{mf}}$ denotes the fixed matrix-free fine-level
elasticity action.  Its operator-level design is treated here as implementation substrate; the
present contribution concerns how that level-0 action is embedded in the GMG
hierarchy.

\paragraph{Level-1 Galerkin assembly.}
$\bK^1 = (\bP^0)^\top \bK^0 \bP^0$ is assembled element-by-element:
\begin{equation}
  K^1_{ij} = \sum_e E_e \bigl(\bP^0_e\bigr)^\top \bm{K}_e \,\bP^0_e
  \bigl[i,j\bigr],
  \label{eq:galerkin_level1}
\end{equation}
where $\bP^0_e\in\mathbb{R}^{24\times 24}$ is the restriction of the global
prolongation to element $e$, computed from trilinear interpolation weights.
Equation~\eqref{eq:galerkin_level1} has complexity $O(n_0)$ and requires
only $E_e$ (no $\bK^0$ materialisation).

Operationally, each coarse element receives contributions from its eight fine
children.  The implementation precomputes those eight local prolongation
patterns, forms the weighted local triple products
$(\bP_e^0)^\top \bm{K}_e \bP_e^0$, gathers the corresponding fine-element
moduli $E_e$, and scatters the resulting values into a deduplicated CSR layout.
The assembled level-1 matrix is therefore obtained without ever building the
global fine stiffness matrix.

\paragraph{Coarser levels ($\ell\ge2$).}
$\bK^\ell = \bR^{\ell-1}\bK^{\ell-1}\bP^{\ell-1}$ via sparse triple products
(cuSPARSE via CuPy).  In the reported 4-level hierarchy on a 216k-element
mesh, the coarsest system has about $1.4\times10^3$ free DOFs and the
coarsest solve is negligible.

\subsection{Chebyshev-Jacobi Smoother}\label{ssec:chebyshev}

Degree-$\nu$ Chebyshev-Jacobi smoothing targets the spectral interval
$[\alpha\lambda_{\max}, \lambda_{\max}]$ of $\bD^{-1}\bK^\ell$, where
$\bD = \mathrm{diag}(\bK^\ell)$ and $\alpha = 1/30$ by default.
The iteration follows the standard Chebyshev semi-iteration described in
Chapter~12 of \citet{saad2003iterative}
with $\bM^{-1} = \bD^{-1}$:

\begin{algorithm}[ht]
\caption{Chebyshev-Jacobi Smoother}\label{alg:cheb}
\begin{algorithmic}[1]
\Require $\bK$, $\bm{b}$, $\bx_0$, $\bD^{-1}$, $\lambda_{\max}$, degree $\nu$, lower fraction $\alpha$
\State $\sigma \leftarrow \tfrac{1}{2}(\lambda_{\max}+\alpha\lambda_{\max})$,\quad
       $\delta \leftarrow \tfrac{1}{2}(\lambda_{\max}-\alpha\lambda_{\max})$
\State $\bm{d}_0 \leftarrow \bD^{-1}(\bm{b}-\bK\bx_0)$,\quad
       $a_0 \leftarrow 2/(\lambda_{\max}+\alpha\lambda_{\max})$
\State $\bx_1 \leftarrow \bx_0 + a_0\,\bm{d}_0$
\For{$k = 1, \dots, \nu-1$}
  \State $\bm{r}_k \leftarrow \bm{b} - \bK\bx_k$
  \State $a_k \leftarrow 1/\bigl(\sigma - \delta^2 a_{k-1}/4\bigr)$
  \State $\bm{d}_k \leftarrow a_k\bigl(\bD^{-1}\bm{r}_k + (\delta^2 a_{k-1}/4)\,\bm{d}_{k-1}\bigr)$
  \State $\bx_{k+1} \leftarrow \bx_k + \bm{d}_k$
\EndFor
\State \Return $\bx_\nu$
\end{algorithmic}
\end{algorithm}

\noindent
The smoothing factor for eigenvalue $\lambda\in[\alpha\lambda_{\max},\lambda_{\max}]$
after $\nu$ steps is bounded by $1/T_\nu(\sigma/\delta)$, where $T_\nu$ is
the Chebyshev polynomial of the first kind.  For $\alpha=1/30$ and $\nu=2$,
$\sigma/\delta=(1+\alpha)/(1-\alpha)=31/29$ and
$1/T_2(31/29)\approx0.78$.  This is only a targeted-band residual-polynomial
bound; it should not be read as a $0.10$ global damping guarantee or as a
measured high-frequency smoothing factor for the full elasticity hierarchy.

\subsection{Precision-Descent Schedule}\label{ssec:prec_descent}

For the BF16 mixed-precision path analysed in this section, we assign
computation precisions per multigrid level as
\begin{equation}
  \text{level } \ell = 0: \text{BF16},\quad
  \text{level } \ell = 1: \text{FP32},\quad
  \text{level } \ell \ge 2: \text{FP64.}
  \label{eq:prec_schedule}
\end{equation}
The outer PCG/FGMRES solver and the residual computation operate in FP64
throughout.  Only the smoother applications and the coarse-grid matvec at
Level~1 use lower precision in that BF16 path.  For the default FP32-GMG runs
reported later, the fine level stays in FP32 and the coarser levels stay in
FP64 unless a dedicated level-precision ablation is stated explicitly.

\paragraph{BF16 numerical semantics.}
Equation~\eqref{eq:prec_schedule} is a level policy, but the implementation is
more specific.  Element moduli are stored in FP64 for the outer solve and
cached in FP32 for lower-precision hierarchy levels.  On the BF16 fine
smoother, the Krylov vector is first cast to FP32; inside the fused kernel, the
local $24$-DOF element tile and the padded $24\times24\rightarrow32\times32$
element tensor are down-cast to BF16 before the WMMA multiply.  Accumulation
remains FP32 on tensor cores, the element modulus $E_e$ is applied in FP32, and
the scattered output also remains FP32 until the smoother returns a correction
promoted back to FP64.  The diagonal inverse is stored in FP64 but cast to FP32
inside the fine smoother; it is never stored in BF16.  No explicit
re-normalisation of $\bD^{-1}\bm{r}$ is required because BF16 shares FP32's
exponent range and the diagonal is floored away from zero during setup.

\paragraph{WMMA kernel placement.}
The tensor-core kernel uses one warp for a $16$-element batch, pads the local
$24$-DOF vectors and stiffness blocks to $32$, and evaluates
$F = U K_e^\top$ with four WMMA $16\times16\times16$ operations per warp.
The launch configuration used in the paper is eight warps per block, hence
$128$ elements per block; shared memory holds the padded BF16 element tensor,
the BF16 displacement batch, the FP32 accumulator tile, and the DOF map.
The resulting shared-memory footprint is approximately $38$\,KiB per block on
Ada, comfortably below the RTX~4090 limit.

The diagnostic motivation is given in Appendix~\ref{app:bound}: in an
idealised frozen SPD setting, the effective condition number
$\kappa_{\mathrm{eff}}$ of the applied left-preconditioned operator
$\mathcal{M}\bK$ can be related to the V-cycle error propagation factor through
\begin{equation}
  \kappa_{\mathrm{eff}} \le
  \frac{1+\rho(\mathcal{E})}{1-\rho(\mathcal{E})},
  \label{eq:kappa_bound}
\end{equation}
where $\mathcal{E}$ is the error propagation operator of the V-cycle and
$\rho(\mathcal{E})<1$ is its spectral radius.
This estimate is used only to motivate the measured
$\eps_{\mathrm{BF16}}\kappa_{\mathrm{eff}}$ screen; it is not a rigorous bound
for the exact floating-point mixed-precision operator.  For the
schedule~\eqref{eq:prec_schedule} the measured
$\kappa_{\mathrm{eff}}$ stays below $256$ on the proxy-compliant subset of the
test suite; Section~\ref{ssec:exp_e5} documents one 216k outlier for which
$\eps_{\mathrm{BF16}}\cdot\kappa_{\mathrm{eff}}>1$.

\paragraph{Spectral-radius estimation.}
Chebyshev smoothing requires an estimate of
$\lambda_{\max}(\bD^{-1}\bK^\ell)$ on each active level.  We use power
iteration: $20$ iterations on the matrix-free fine level and $10$ on each
assembled coarse level.  The estimate is cached across solves and recomputed
only when $\max_e E_e$ changes by more than $10\%$, which serves as a cheap
proxy for substantial SIMP-state drift.  These power iterations are part of the
reported setup time in Section~\ref{ssec:exp_e7}; there is no tolerance-based
inner stopping rule.  This smoother-tuning estimate is distinct from the
Lanczos-based $\kappa_{\mathrm{eff}}$ probe used later in the heterogeneous
spectral-proxy and robustness screens.

\subsection{Outer Krylov Solver and Convergence Statement}

For the FP64/FP32 comparison paths, the outer solve uses standard
left-preconditioned CG with the frozen V-cycle defining the applied operator
$\mathcal{M}$.  Within the implemented solver's default configuration, the
automatic outer-solver policy selects PCG for FP64/FP32 fine smoothers and
restarted right-preconditioned FGMRES for the BF16 path.  The
$\kappa_{\mathrm{eff}}$ diagnostic and the conditional spectrum bound below
therefore refer specifically to the frozen left-preconditioned operator
$\mathcal{M}\bK$ used by the FP64/FP32 PCG pairings.  A supplementary
small-problem diagnostic
(listed in Table~\ref{tab:artifact_modules}) records a nonzero symmetry defect
for the applied floating-point preconditioned operator
($0.179$ in the max-absolute diagnostic), so the FP32/FP64 PCG rows in this
paper should be read as empirical solver pairings on the reported comparison
paths rather than as a proof that the implemented V-cycle is exactly SPD in
floating point.  The implementation also provides explicit
FGMRES paths whenever restart sensitivity is being screened or the active
benchmark path is treated conservatively as non-symmetric; this includes the
sensitivity surface, the large-scale FP32 solves, and the robustness screens.
The implemented FGMRES variant uses modified
Gram--Schmidt Arnoldi orthogonalization, Givens rotations for the incremental
QR update, and a happy-breakdown check when the new Hessenberg entry drops
below $10^{-14}$.  The default FGMRES path uses restart $32$; the E6
sensitivity surface screens $16/32/50$, and the large-scale and robustness
stress tests use
restart $50$ with the larger outer-iteration caps stated in the experiment
sections.
Unless noted otherwise, convergence means
$\|\bm{r}_k\|_2/\|\bm{b}\|_2 < 10^{-6}$.
Operationally, reported solves are accepted only by this true-residual check.
Iteration caps, stagnant residual histories, or non-finite iterates are
recorded as failures; the $\eps_{\mathrm{BF16}}\kappa_{\mathrm{eff}}$ scalar is
never used by itself to accept a BF16 result.

\begin{remark}[Conditional diagnostic for the frozen hierarchy]
\label{thm:fgmres_gmg}
Consider one linear solve at a fixed SIMP state, so that $\{E_e\}$ and hence
the left-preconditioning operator $\mathcal{M}$ induced by one frozen V-cycle
are fixed during the Krylov iteration.  Assume:
\textup{(A1)} each coarse operator is SPD on the free-DOF subspace because the
Dirichlet mask is inherited by injection and the hierarchy is Galerkin;
\textup{(A2)} the smoother and coarse correction give an error-propagation
operator $\mathcal{E}$ with $\rho(\mathcal{E})<1$; and
\textup{(A3)} floating-point perturbations are small enough that the frozen
operator can be interpreted through the same SPD spectral picture.
Under these idealisations, the effective condition number estimate
\begin{equation}
  \kappa_{\mathrm{eff}}
  \le
  \frac{1+\rho(\mathcal{E})}{1-\rho(\mathcal{E})}.
  \label{eq:theorem_kappa}
\end{equation}
is the quantity we use to interpret the Lanczos probe.  The quantity
$\eps_{\mathrm{BF16}}\kappa_{\mathrm{eff}}<1$ is therefore an empirical
spectral proxy on measured cases, not an assumption-free theorem or
convergence classifier for the implemented mixed-precision V-cycle.
\end{remark}

\subsection{Complexity and Memory Model}

\begin{table}[!tbp]
\centering
\small
\caption{Per-level cost model for one operator/smoother application inside a
V-cycle.  $n_\ell$ denotes free DOFs, $n_{\mathrm{elem},\ell}$ elements,
$\mathrm{nnz}_\ell=\mathrm{nnz}(\bK^\ell)$, and $s_v$ the value width in bytes
(4 for FP32, 8 for FP64).}
\label{tab:cost_model}
\begin{tabular}{p{0.10\linewidth}p{0.24\linewidth}p{0.23\linewidth}p{0.28\linewidth}}
\toprule
Level & Stored operator data & FLOPs per application & Dominant bytes touched \\
\midrule
$\ell=0$ & Element moduli $E_e$, diagonal $\bD^{-1}$, transfer operators $\bP^0,\bR^0$ & Fine matvec $\approx 2\cdot24^2 n_{\mathrm{elem},0}$; degree-$\nu$ smoother costs $\nu$ such matvecs plus $O(\nu n_0)$ vector updates & Gather/scatter traffic on element DOFs plus vector reads/writes; no finest-level CSR matrix \\
$\ell=1$ & FP32/FP64 CSR $\bK^1$, diagonal $\bD^{-1}_1$, $\bP^1,\bR^1$ & SpMV $\approx 2\,\mathrm{nnz}_1$; smoother adds $O(n_1)$ diagonal scaling and saxpy work & $(s_v+4)\mathrm{nnz}_1 + s_v(n_1+1)$ for CSR data/indices/indptr, plus vectors \\
$\ell\ge2$ & FP64 CSR $\bK^\ell$, diagonal $\bD^{-1}_\ell$, $\bP^\ell,\bR^\ell$ & SpMV $\approx 2\,\mathrm{nnz}_\ell$; coarsest solve is dense Cholesky if $n_L\le5000$, else fixed-count PCG & CSR traffic dominates until the coarsest level; dense solve cost is negligible for the reported $n_L$ \\
\bottomrule
\end{tabular}
\end{table}

Under structured $2{:}1$ coarsening, $n_{\ell+1}\approx n_\ell/8$ away from
boundaries, so the V-cycle cost is dominated by the finest level.  This is why
the precision policy is intentionally asymmetric: only the most expensive level
is pushed to BF16, the first assembled coarse level is kept in FP32 to
stabilise the first correction, and the deeper levels revert to FP64 because
their absolute cost is already small.

\begin{figure}[!tbp]
\centering
\includegraphics[width=\linewidth]{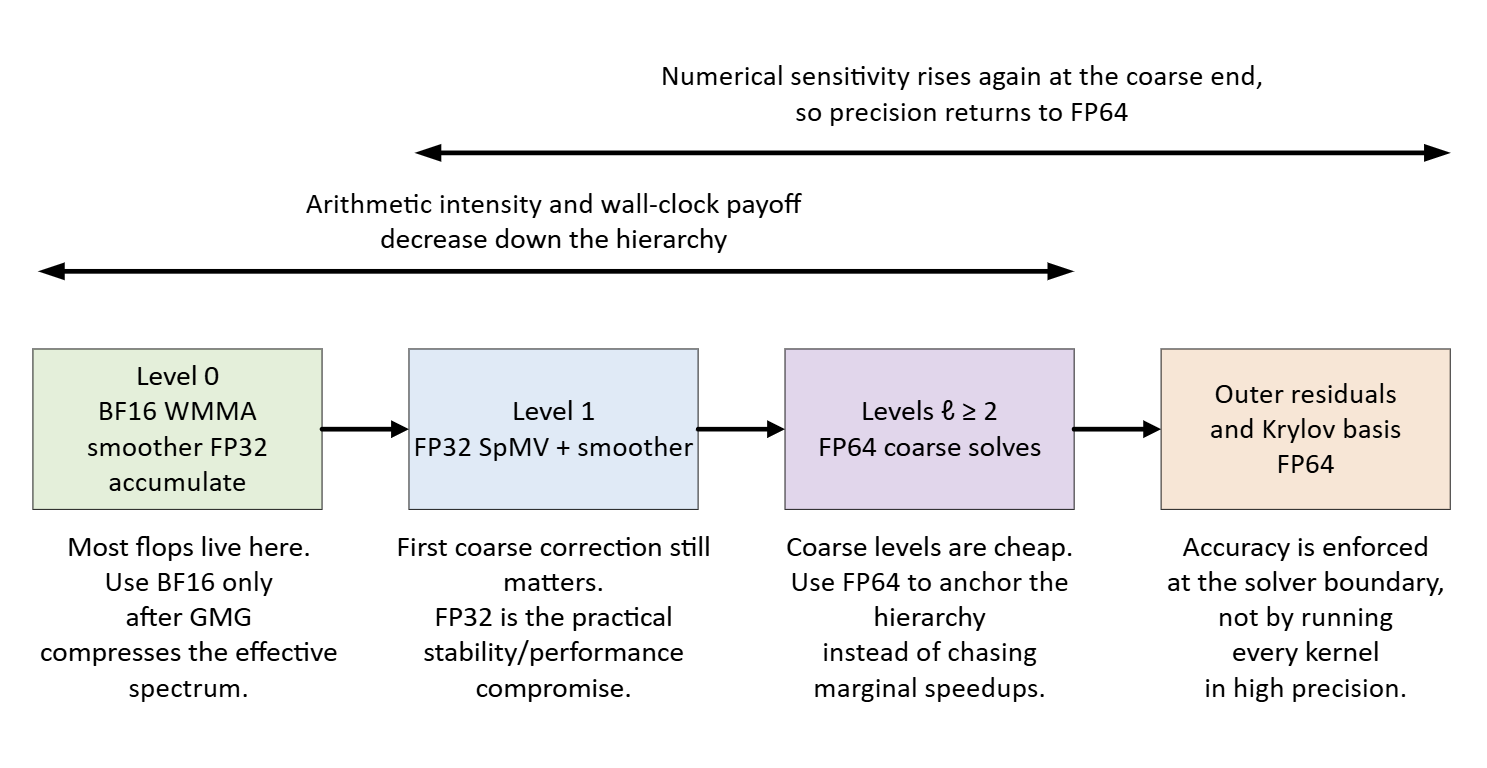}
\caption{Precision descent across the hierarchy.
  The expensive operation is the fine-level smoother, so that is where BF16
  tensor cores are used; the first assembled coarse level rises to FP32 to
  stabilise the first correction; deeper levels return to FP64 because their
  marginal runtime cost is tiny and they anchor the hierarchy numerically.
  The outer residual and Krylov basis stay in FP64 throughout.
  The figure makes the precision logic explicit before the reader reaches the
  measured throughput and convergence sections.}
\label{fig:precision_overview}
\end{figure}

\subsection{Coarse-Operator Choice: Why Galerkin Rather Than Re-discretization}
\label{ssec:galerkin_vs_rediscr}

Re-discretization coarsens the density field $\bm\rho$ and reassembles $\bK^1$
from the coarsened density.
This decouples coarse and fine operators, which can reduce convergence rates
when the fine-level material distribution changes rapidly.
Galerkin coarsening satisfies the \emph{variational condition}
$\bK^1 = (\bP^0)^\top\bK^0\bP^0$ by construction, guaranteeing that the
coarse space is the optimal Galerkin approximation of the fine-level
correction problem.  In the filtered SIMP setting of \eqref{eq:simp_filter},
this also means that all coarse operators inherit the already filtered modulus
field implicitly through the fine operator; no separate coarse-grid filter or
projection is applied inside the hierarchy.
Validation check M3 in Section~\ref{ssec:phase1} therefore compares the
matrix-free Galerkin hierarchy against an assembled-Galerkin reference on a
small problem; the re-discretization discussion here is methodological context
rather than a separate benchmark result.

\section{Experimental Setup and Results}\label{sec:results}

All experiments reported below were run in the reported reproduction environment
on a single NVIDIA RTX~4090.
The numbered experiment groups are summarized in
Table~\ref{tab:experiment_map}.
Repeated-run quantities in E2, E6, E7, E8, and E9 are reported as
mean$\pm$standard deviation over $10$ timed trials after $2$ warm-up runs of
the corresponding solve path.  E4 is a dedicated 200-repetition kernel proxy
after $5$ warm-up launches, E3 is one auxiliary 30-step continuation run per
benchmark, and the E6 sensitivity surface uses a lighter $1$-warm-up,
$3$-trial screening protocol.  The direct BF16 validation and the
high-contrast smoother screen are single-run diagnostic sweeps with explicit
true-residual checks.  Full residual histories, timing sidecars, trajectory
diagnostics, and screening-map records are retained in the reproducibility
package described in Appendix~\ref{app:repro}.

\subsection{Experimental Setup and Solver Configuration}\label{ssec:exp_setup}

The solver stack evaluated below is the concrete implementation of the
methodology in Section~\ref{sec:method}, including the padded
$24\rightarrow32$ WMMA fine-level kernel and the restart policy for the outer
Krylov solver.  For BF16 runs the first assembled coarse level is kept in FP32;
the default FP32-GMG runs instead keep that level in FP64 unless a dedicated
level-precision ablation is reported explicitly.
For the representative 216k-element case, the first assembled coarse level has
about $89$k free unknowns, so its cost is small relative to the matrix-free
fine-level action.
The full hierarchy---fine operator, Galerkin assembly, coarse SpMV, and
outer Krylov loop---is implemented in CuPy/Python with runtime kernel
compilation.
All reported timings use the pinned hardware/software environment and fixed
seeds summarized in Appendix~\ref{app:repro}.  The benchmark stack is
Python~3.10.18, CuPy~13.6.0 (\texttt{cupy-cuda12x}), NumPy~2.2.6,
SciPy~1.15.3, Matplotlib~3.10.7, and PyAMG~5.3.0; the rendering stack uses
PyVista~0.46.3 and VTK~9.6.0.
The live workstation used for the reported runs was an RTX~4090 with NVIDIA
driver~595.71; \texttt{nvidia-smi} reported CUDA API support 13.2, while the
execution backend was CuPy's \texttt{cupy-cuda12x}~13.6.0 CUDA~12.x runtime
package.
Unless stated otherwise, repeated-run measurements use two warm-up runs and ten
timed trials; the heterogeneous probes use experiment-specific fixed seeds
listed in Appendix~\ref{app:repro} (seed~42 for the E1/E5 binary-contrast
sweeps).
The reported implementation also applies a small set of fixed
numerical safeguards:
the power-iteration estimate of $\lambda_{\max}$ is floored at $10^{-6}$, the
coarsest-level regularization is
$\varepsilon=\max(\text{mean\_diag}\cdot10^{-8},10^{-14})$, and the small
FGMRES least-squares system falls back to \texttt{lstsq} if the triangular
solve becomes singular.
Unless an experiment overrides them explicitly, the default hierarchy uses four
levels and the smoothing configuration summarized in
Table~\ref{tab:repro_defaults}.  That table also lists the solver defaults and
fixed experiment seeds needed to reproduce the reported runs.

The experiments answer four distinct questions.
First, how far does the hierarchy mitigate the mesh- and contrast-driven
stagnation seen in Jacobi-PCG?
Second, where does the speedup come from: fewer outer iterations, faster
fine-level arithmetic, or both?
Third, how informative is a spectral proxy once it is checked
against direct BF16 solves?
Fourth, how far does the single-GPU regime extend before setup cost or VRAM
becomes the dominant constraint?

Table~\ref{tab:experiment_map} separates the density construction and evidence
role of each experiment.  This separation is important because the benchmark
suite mixes frozen linear systems, proxy kernel measurements, a limited
auxiliary OC schedule, and robustness screens; only E3 updates a design over
multiple SIMP steps.

\begin{table}[!tbp]
\centering
\caption{Experiment map and density construction.  Here
  $\rho_{\mathrm{floor,test}}$ is the synthetic density-construction floor, while
  $E_{\min}=10^{-9}$ is the SIMP modulus floor used in the reported runs.}
\label{tab:experiment_map}
\footnotesize
\begin{tabularx}{\linewidth}{l>{\RaggedRight\arraybackslash}p{0.24\linewidth}>{\RaggedRight\arraybackslash}p{0.27\linewidth}>{\RaggedRight\arraybackslash}X}
\toprule
ID & State construction & Main solver path & Evidence role\\
\midrule
E1 & Fixed-seed binary-contrast fields, $V_f\in\{0.2,0.5,0.8\}$, $p\in\{1.5,3,4.5\}$, $\rho_{\mathrm{floor,test}}=10^{-2}$ & FP64-GMG, outer PCG, 200 cap & Heterogeneous stress sweep and failure-rate map\\
E2 & Uniform $\rho=0.5$, $p=3$ at 64k--512k & Jacobi-PCG, FP32-GMG, BF16-GMG & Direct per-solve timing and residual histories\\
E3 & Auxiliary 30-step fixed-penalty OC schedule, $p=3$, tolerance $10^{-5}$ & Jacobi-PCG vs.\ FP32-GMG & Schedule-execution timing, not matched final design\\
E4 & Uniform $E_e=1$ proxy input on 216k & BF16/FP32 fine matvec proxy & Kernel-throughput context only\\
E5 & E1-style fixed-seed fields at 64k and 216k & Lanczos probe plus direct BF16/FP64 FGMRES checks & BF16 spectral proxy and direct convergence validation\\
E6 & Uniform 216k state plus selected high-contrast fields & FP64/FP32/BF16 ablations & Precision, depth, cycle, smoother, restart, and high-contrast smoother sensitivity\\
E7 & Uniform modulus $E_e=0.5$ at 125k--1M & FP32-GMG with FGMRES, restart 50 & Large-scale uniform linear solves and setup VRAM\\
E8 & Uniform 64k state after CPU CSR assembly & CPU PyAMG vs.\ GPU FP32-GMG & Narrow post-assembly external baseline\\
E9 & Uniform 216k state & FP64-GMG and FP32-GMG & NVIDIA Management Library (NVML) solve-window energy proxy\\
E10 & Fixed-seed and deterministic stress fields at 64k & FP64-GMG with FGMRES & Robustness/failure-mode screen\\
\bottomrule
\end{tabularx}
\end{table}

\begin{table}[!tbp]
\centering
\caption{Outer-solver policy used in the reported experiment groups.  PCG rows
  are empirical solver pairings, not SPD certificates for the floating-point
  V-cycle; FGMRES rows are the conservative non-symmetric paths.}
\label{tab:solver_policy}
\footnotesize
\begin{tabularx}{\linewidth}{l>{\RaggedRight\arraybackslash}p{0.19\linewidth}>{\RaggedRight\arraybackslash}p{0.23\linewidth}rr>{\RaggedRight\arraybackslash}X}
\toprule
Group & Path & Outer solver & Cap & Tol. & Interpretation\\
\midrule
E1 & FP64-GMG heterogeneous sweep & V-cycle-preconditioned CG & 200 & $10^{-6}$ & Failure-rate stress map\\
E2 & Jacobi-PCG and FP32-GMG & Preconditioned CG & 200 & $10^{-6}$ & Timed uniform-state comparison; Jacobi path is capped and non-converged\\
E2 & BF16-GMG & FGMRES, restart 32 & 200 & $10^{-6}$ & Guarded mixed-precision comparison\\
E3 & Auxiliary OC schedule & Solver defaults, cap 1000 & 1000 & $10^{-5}$ & Auxiliary schedule diagnostic only\\
E5 direct & FP64/BF16 heterogeneous validation & FGMRES, restart 50 & 500 & $10^{-6}$ & Direct true-residual BF16 check\\
E6 & Ablations and sensitivity screens & PCG or FGMRES as labelled; restarts 16/32/50 in FGMRES screen & 200 & $10^{-6}$ & Solver-policy sensitivity\\
E6 high contrast & High-contrast smoother screen & FGMRES, restart 50 & 500 & $10^{-6}$ & Smoother failure diagnostic\\
E7 & Large-scale uniform solves & FGMRES, restart 50 & 500 & $10^{-6}$ & Uniform-modulus scaling only\\
E8 & PyAMG / FP32-GMG external baseline & PyAMG smoothed-aggregation AMG+CG / GMG-PCG & 200 & $10^{-6}$ & Narrow post-assembly 64k reference\\
E10 & Robustness and basin screens & FGMRES, restart 50 & 500 / 300 & $10^{-6}$ & Failure-mode screen\\
\bottomrule
\end{tabularx}
\end{table}

\subsection{Validation Checks Before Main Experiments}\label{ssec:phase1}

Before reporting wall-clock or scaling claims, we verify that the hierarchy
behaves correctly on a compact set of sanity checks: direct-solve agreement,
bounded uniform-density iteration counts, matrix-free/assembled-Galerkin
agreement, smoother behavior, selected SIMP sanity probes, the
$\kappa_{\mathrm{eff}}$ bound, and mixed-precision compliance parity.
Table~\ref{tab:phase1} summarises these checks.

\begin{table}[!tbp]
\centering
\caption{Validation checks completed before the main experiments.}
\label{tab:phase1}
\small
\begin{tabularx}{\linewidth}{c>{\RaggedRight\arraybackslash}p{0.30\linewidth}>{\RaggedRight\arraybackslash}p{0.22\linewidth}>{\RaggedRight\arraybackslash}Xc}
\toprule
ID & Description & Criterion & Retained value & Status\\
\midrule
M1 & FP64 V-cycle vs.\ direct solve (64k) & residual $<10^{-10}$ & relative residual $5.94{\times}10^{-11}$ & \checkmark\\
M2 & bounded iteration count (64k--512k, uniform density) & FGMRES iters $\le 30$ & 17, 21, and 13 iterations & \checkmark\\
M3 & Matrix-free vs.\ assembled Galerkin & compliance rel.\ diff.\ $<0.1\%$ & relative difference $1.19{\times}10^{-15}$ & \checkmark\\
M4 & Chebyshev degree-2/degree-4 vs.\ Jacobi smoother & converges with iters $\le 50$ & 27, 30, and 21 iterations & \checkmark\\
M5 & Selected SIMP sanity probes ($p\in\{1.5,3,4.5\}$, $E_{\min}=10^{-9}$) & converges on validation probes & 17, 17, 17, and 174 iterations & \checkmark\\
M6 & $\kappa_{\mathrm{eff}}\le256$ on the nominal 64k, $p=3$, $\rho=0.5$ probe & Lanczos-based spectral probe & $\kappa_{\mathrm{eff}}=79.65$ & \checkmark\\
M7 & BF16 drop-in compliance error & $\le0.5\%$ vs.\ FP64; FGMRES restart 50, maxiter 2000; compliance gate only & BF16 error $3.66{\times}10^{-6}\%$; 29 iterations, converged & \checkmark\\
M8 & Three-level FP32 hierarchy on all 4 benchmarks & compliance error $\le0.5\%$ & maximum error $3.08{\times}10^{-8}\%$ & \checkmark\\
\bottomrule
\end{tabularx}
\end{table}

M2 is intentionally a uniform-density gate; the heterogeneous E1 stress test
reported below is harder.  M6 is likewise a single-point probe, while the
broader spectral-proxy sweep appears later in Experiment~E5.
These validation checks are pre-benchmark sanity gates, not robustness
certificates; the heterogeneous failure-mode screen is reported separately in
Experiments~E1 and E10.
Table~\ref{tab:phase1} reports the measured validation margins rather than only
pass/fail status; Appendix~\ref{app:artifact} gives the corresponding artifact
trace.

\subsection{Outer Iteration Count vs.\ Mesh Size}\label{ssec:exp_e1}

Figure~\ref{fig:e1} reports the outer iteration count across the full
27-case sweep over mesh size, volume fraction $V_f\in\{0.2,0.5,0.8\}$ of a
heterogeneous binary-contrast density field (solid/void voxels placed by a
fixed pseudorandom seed, $\rho_{\mathrm{floor,test}}=10^{-2}$), and SIMP penalization
$p\in\{1.5,3,4.5\}$.
Validation check M2 uses FGMRES on the uniform-density probe as a conservative
gate, whereas E1 reports the empirical PCG comparison path on 27 harder
heterogeneous states of the same FP64 hierarchy.
The first statistic is pass rate, not the capped-inclusive mean.  Under the
200-iteration budget, 7/9, 4/9, and 1/9 cases converge at 64k, 216k, and 512k,
respectively.  The converged-only mean iteration counts are therefore about
112, 134, and 146; these numbers describe the successful subset only.
Across all cases, including capped failures, the size-wise means are 131, 170,
and 194 iterations.  Those capped-inclusive means are retained in the figure as
cap-aware stress-test summaries, not as converged performance metrics.
Iteration counts are therefore benchmark-dependent rather than strictly
$h$-independent.  The separate uniform-density Jacobi baseline in
Experiment~E2 still hits the 200-iteration cap at all three mesh sizes, but it
is kept visually separate because it is not the same heterogeneous 27-case
sweep.  At 512k, eight of the nine heterogeneous-field configurations do not
converge within the 200-iteration budget, so the 194-iteration mean should be
read as a failure-dominated stress-test outcome rather than as evidence of
stable mesh-independent convergence.

\begin{figure}[!tbp]
\centering
\includegraphics[width=0.72\linewidth]{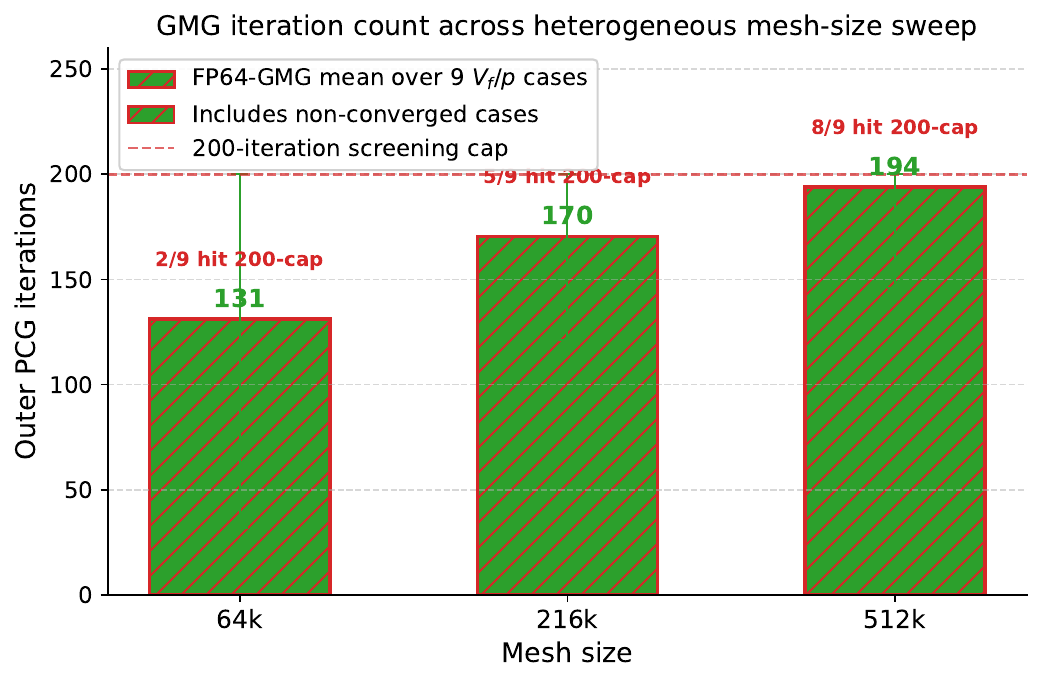}
\caption{E1: FP64-GMG outer iteration count vs.\ mesh size on the
  heterogeneous 27-case cantilever sweep (outer PCG).  Bars show the size-wise mean over
  the nine $(V_f,p)$ combinations with min--max error bars; hatched bars
  include cases that hit the 200-iteration screening cap, and the dashed line
  marks that cap.  Failure rates are 2/9 at 64k, 5/9 at 216k, and 8/9 at 512k.}
\label{fig:e1}
\end{figure}

\subsection{Per-Linear-Solve Wall Time}\label{ssec:exp_e2}

Figure~\ref{fig:e2a} shows wall-time scaling.
Figure~\ref{fig:e2b} reports the capped-baseline wall-time ratios, again as
mean$\pm$standard deviation over ten trials.
FP32-GMG is faster than the capped Jacobi-PCG path at all three sizes, with
mean wall-time ratios of
$1.62\times$, $1.75\times$, and $3.12\times$ at 64k, 216k, and 512k,
respectively.  These ratios are not speedups over successful Jacobi-PCG
solves: the Jacobi-PCG path hits the 200-iteration cap without convergence in
all ten timed trials at each size.  BF16-GMG remains slower than FP32-GMG at
every size.  It is not uniformly slower than the capped Jacobi path either:
at 216k it achieves a $1.06\times$ speedup and at 512k a $2.60\times$ speedup
over the capped baseline, while requiring more outer iterations than the FP32
hierarchy.  Figure~\ref{fig:e2c} shows the corresponding change in solver behavior:
the GMG variants reduce the residual by orders of magnitude in a few tens of
outer iterations, whereas Jacobi-PCG stagnates at the 200-iteration cap.
For these E2 solves, the Jacobi-PCG baseline and FP32-GMG use PCG, whereas
BF16-GMG uses restarted FGMRES under the solver policy used in the
implementation; the BF16 iteration counts reported below are mean FGMRES
outer iterations across trials.

\begin{figure}[!tbp]
\centering
\includegraphics[width=0.72\linewidth]{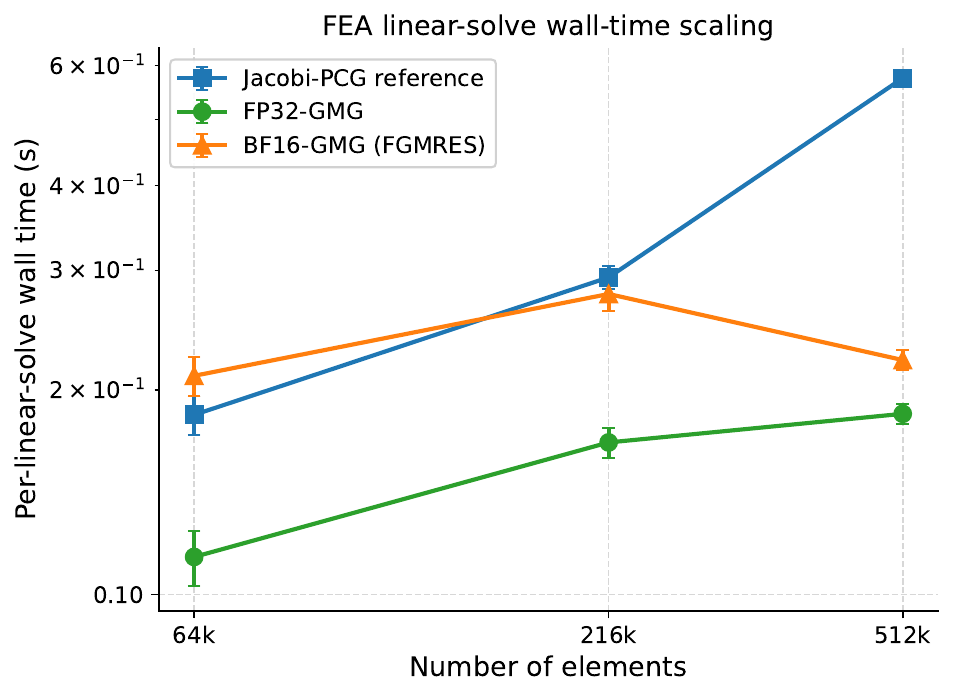}
\caption{Per-solve wall-time scaling for the 64k/216k/512k
  $\rho=0.5$, $p=3.0$ linear solves comparing the Jacobi-PCG baseline of
  \citet{yang2026fused}, FP32-GMG, and
  BF16-GMG (FGMRES outer solver).  Markers/lines report the mean over ten timed trials; error bars show one
  standard deviation.}
\label{fig:e2a}
\end{figure}

\begin{figure}[!tbp]
\centering
\includegraphics[width=0.72\linewidth]{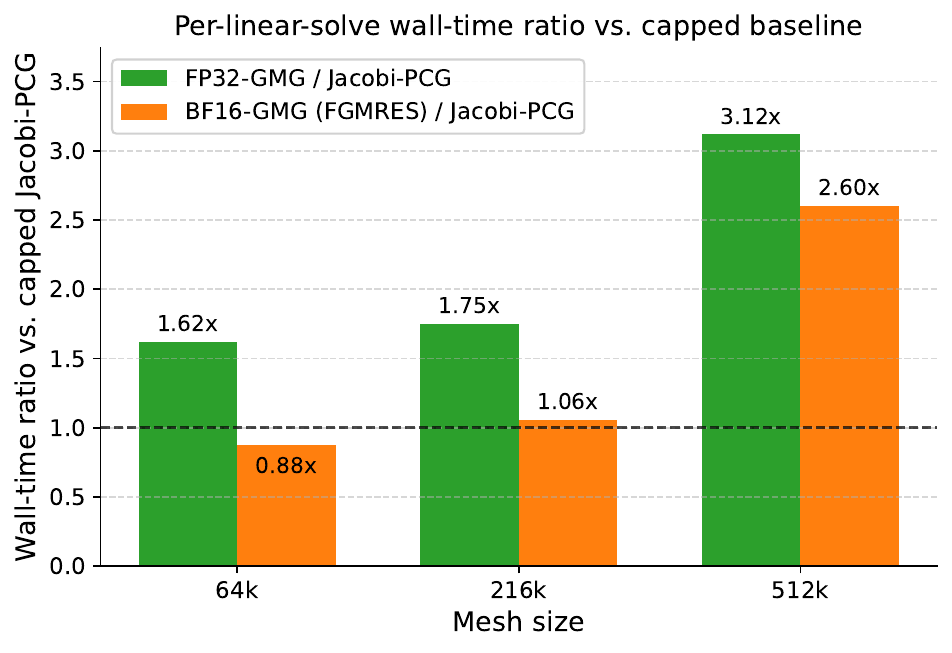}
\caption{Per-solve wall-time ratio relative to the capped Jacobi-PCG path on
  the same $\rho=0.5$, $p=3.0$ linear solves.  The Jacobi-PCG path reaches the
  200-iteration cap without convergence in all timed trials; bars compare
  FP32-GMG and BF16-GMG (FGMRES outer solver);
  error bars reflect the corresponding timed-trial variation.}
\label{fig:e2b}
\end{figure}

\begin{figure}[!tbp]
\centering
\includegraphics[width=\linewidth]{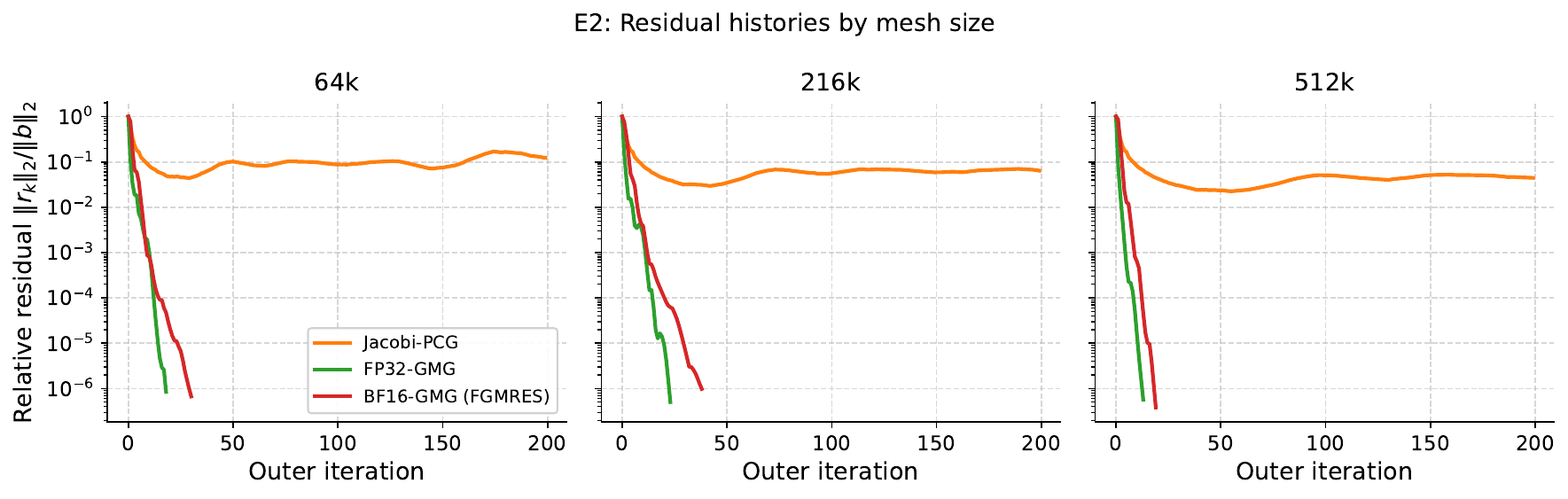}
\caption{Residual histories at 64k, 216k, and 512k.
  GMG changes the outer-solver regime rather than merely shaving a constant
  factor from the same stagnating Krylov path.  The BF16-GMG traces use
  restarted FGMRES; the Jacobi-PCG baseline and FP32-GMG traces use PCG.}
\label{fig:e2c}
\end{figure}

\begin{table}[!tbp]
\centering
\caption{E2: Per-linear-solve wall time, iterations, and capped-baseline
  wall-time ratio
  ($\rho=0.5$, $p=3.0$).  The BF16-GMG column uses restarted FGMRES; the
  Jacobi and FP32-GMG columns use PCG.  Jacobi-PCG reaches the 200-iteration
  cap without convergence in all ten timed trials at every size, so the
  reported ratios compare against a capped non-converged reference path.
  BF16-GMG uses FGMRES and shows small restart-level iteration-count variation
  at 64k and 216k; at 512k, all ten BF16 trials converge in exactly
  19 iterations, so the zero iteration-count standard deviation is not trial
  filtering.}
\label{tab:e2}
\footnotesize
\setlength{\tabcolsep}{2pt}
\begin{tabular}{lrrrrrrrr}
\toprule
Size & Jacobi (s) & Jacobi it. (cap) & FP32 (s) & FP32 it. & BF16 (s) & BF16 it. & FP32/Jac. & BF16/Jac.\\
\midrule
64k  & $0.184\pm0.012$ & 200 & $0.114\pm0.010$ & 18 & $0.210\pm0.014$ & $29.3\pm0.5$ & $1.62\times$ & $0.88\times$\\
216k & $0.293\pm0.012$ & 200 & $0.167\pm0.009$ & 23 & $0.277\pm0.015$ & $38.5\pm1.4$ & $1.75\times$ & $1.06\times$\\
512k & $0.575\pm0.009$ & 200 & $0.184\pm0.006$ & 13 & $0.221\pm0.008$ & $19.0\pm0.0$ & $3.12\times$ & $2.60\times$\\
\bottomrule
\end{tabular}
\end{table}

\paragraph{Auxiliary OC schedule diagnostic.}
The auxiliary 30-step fixed-penalty OC schedule is not used as main acceleration
evidence because the design trajectories diverge and the baseline repeatedly
hits the iteration cap.  For transparency, the same-schedule timing and
trajectory plots are retained in Appendix~\ref{app:e3_aux}.

\subsection{Fine-Operator Throughput Context}\label{ssec:exp_e4}

\begin{figure}[!tbp]
\centering
\includegraphics[width=0.52\linewidth]{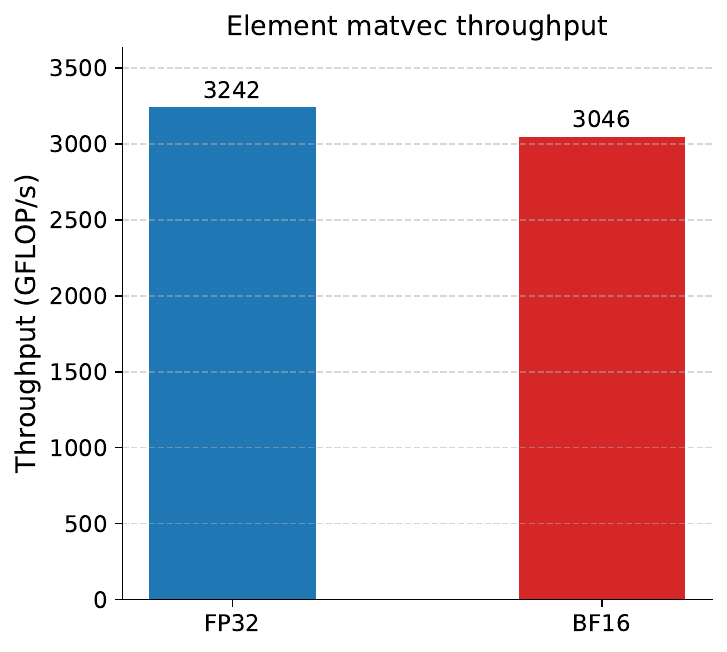}
\caption{Fine-level operator proxy throughput on RTX~4090 (200-rep benchmark).
  BF16 WMMA tensor-core path vs.\ FP32 scalar CUDA cores.
  This proxy is included only to interpret the solver-level mixed-precision
  results; it is not a separate fused-operator benchmark claim of this paper.}
\label{fig:e4}
\end{figure}

The BF16 WMMA fused kernel achieves 3\,046\,GFLOP/s on the RTX~4090
(200-repetition benchmark on the representative 216k cantilever case with
uniform element modulus $E_e=1$ and an all-ones free-DOF input vector,
Figure~\ref{fig:e4}),
compared to 3\,242\,GFLOP/s for the equivalent FP32 fused kernel.
These proxy timings are reported only to interpret the solver-level
mixed-precision behavior of the GMG hierarchy.
On this proxy benchmark the two paths are effectively tied, so BF16 is not a
practical speed path in the present Q1-hex implementation.
Figure~\ref{fig:e4b} places that proxy measurement in a simple illustrative
roofline model whose reference lines use vendor-spec RTX~4090 guide lines
(FP32 compute, BF16 tensor-core compute, and DRAM bandwidth) rather than
measured saturation ceilings \citep{nvidia2023ada}.  The fine matrix-free
kernel sits well below the BF16 compute ceiling and close to the bandwidth
roof implied by its low arithmetic intensity; level-1 SpMV is even more
strongly bandwidth-limited, while the coarsest dense solve is too small to
matter for end-to-end wall time.

\begin{figure}[!tbp]
\centering
\includegraphics[width=0.72\linewidth]{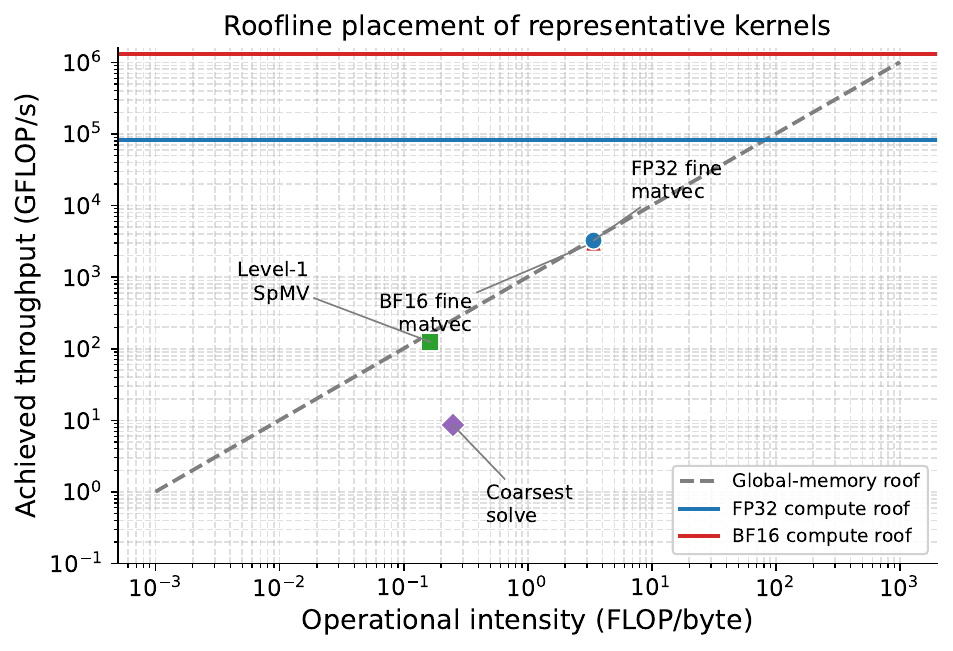}
\caption{Illustrative roofline placement of representative kernels on the
  RTX~4090 using vendor-spec reference guide lines \citep{nvidia2023ada}.
  The fine matrix-free smoother is bandwidth-limited rather than tensor-core
  compute-limited, which is consistent with the observed lack of a large
  end-to-end BF16 throughput gain despite using WMMA.}
\label{fig:e4b}
\end{figure}

\subsection[Empirical BF16 spectral-proxy map]{Empirical BF16 spectral-proxy map:
  $\varepsilon_{\mathrm{BF16}}\!\cdot\!\kappa_{\mathrm{eff}}$}
\label{ssec:exp_e5}

Figure~\ref{fig:e5} plots
$\varepsilon_{\mathrm{BF16}}\cdot\kappa_{\mathrm{eff}}$ as a function of
volume fraction $V_f\in\{0.2,0.5,0.8\}$ of a heterogeneous
binary-contrast density field, SIMP penalization
$p\in\{1.5,3,4.5\}$, mesh size (64k, 216k), and the synthetic test floor
$\rho_{\mathrm{floor,test}}=10^{-2}$, estimated via a Lanczos-based spectral probe of the
applied left-preconditioned operator on the FP64 hierarchy used for the
heterogeneous probe.  Each grid point is one fixed-seed Bernoulli
realization (seed 42), so this figure should be read as a screening probe
rather than as a frequency estimate across realizations.  We use this scalar as
a proxy diagnostic for the mixed-precision schedule, not as a certified BF16
convergence or stability test.

The direct BF16 validation runs every state in this spectral-proxy grid with
FGMRES, a 500-iteration cap, and FP64 true-residual and compliance checks
against the same frozen state.  Table~\ref{tab:e5_direct} shows that the proxy is not a
convergence classifier.  Seventeen of the eighteen tested combinations satisfy
$\varepsilon_{\mathrm{BF16}}\cdot\kappa_{\mathrm{eff}} < 1$, but only six of
those screened-in cases converge in BF16; eleven screened-in cases still reach
the BF16 iteration cap.  The single screened-out case is the near-void 216k state
($V_f=0.2$, $p=1.5$) with
$\varepsilon_{\mathrm{BF16}}\cdot\kappa_{\mathrm{eff}}\approx 12.59$, and it
  converges in 154 BF16 iterations.  Overall, BF16 and FP64 both converge on
  7/18 tested states.  The largest BF16--FP64 compliance difference is only
  0.0063\%, but this parity should not be read as convergence when the true
  residual remains above tolerance.  For the eleven screened-in capped states,
  FP64 and BF16 final true residuals both remain above tolerance and lie in the
  same broad residual range ($5.5{\times}10^{-6}$--$3.7{\times}10^{-2}$),
  supporting the interpretation that these rows expose hierarchy/convergence
  stagnation rather than large BF16--FP64 compliance drift.

\begin{figure}[!tbp]
\centering
\includegraphics[width=0.72\linewidth]{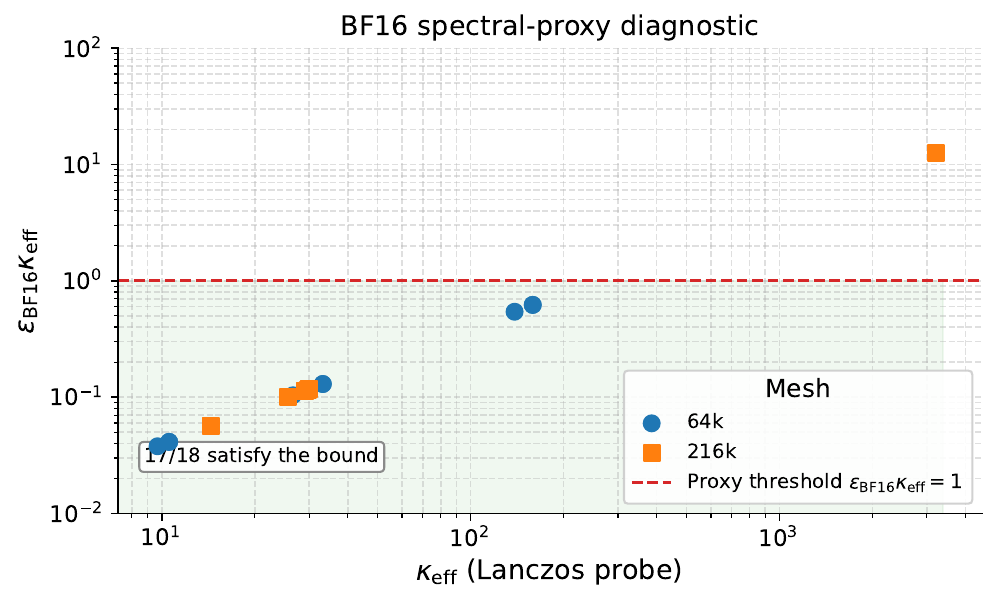}
\caption{$\varepsilon_{\mathrm{BF16}}\!\cdot\!\kappa_{\mathrm{eff}}$ at the
  fine-system linear operator under the full hierarchy across all 18
  $(\text{size},V_f,p)$ combinations from the
  fixed-seed heterogeneous screening probe (seed 42), estimated via Lanczos
  on the frozen FP64 hierarchy as a spectral proxy.  The
  heterogeneous sweep uses the synthetic test floor
  $\rho_{\mathrm{floor,test}}=10^{-2}$.
  Seventeen points lie below the dashed proxy threshold; one 216k
  outlier lies above it and is shown without filtering.  Both axes are
  logarithmic.}
\label{fig:e5}
\end{figure}

\FloatBarrier

\begin{table}[!htbp]
\centering
\caption{Direct BF16 validation on the 18 fixed-seed heterogeneous spectral-proxy states.
  Convergence uses the FP64 true residual tolerance $10^{-6}$ with a
  500-iteration FGMRES cap.  The proxy screen is informative but not a
  pass/fail classifier for the implemented mixed-precision hierarchy.}
\label{tab:e5_direct}
\small
\begin{tabular}{lrrrr>{\RaggedRight\arraybackslash}p{0.31\linewidth}}
\toprule
Screen result & Cases & BF16 conv. & BF16 capped & FP64 conv. & Interpretation\\
\midrule
$\eps_{\mathrm{BF16}}\kappa_{\mathrm{eff}}<1$ & 17 & 6 & 11 & 6 & Eleven screened-in cases still cap; paired true-residual stagnation, not large BF16--FP64 compliance drift\\
$\eps_{\mathrm{BF16}}\kappa_{\mathrm{eff}}>1$ & 1 & 1 & 0 & 1 & One screened-out case at 216k, $V_f=0.2$, $p=1.5$; BF16 converges in 154 iterations\\
\midrule
Total & 18 & 7 & 11 & 7 & Max BF16--FP64 compliance difference: 0.0063\%\\
\bottomrule
\end{tabular}
\end{table}

\subsection{Ablations}\label{ssec:exp_e6}

Each ablation is measured in a separate experimental run.
Small wall-time differences between nominally identical configurations across
subfigures should therefore be read as ordinary inter-run GPU variation rather
than as a single factorial timing table.

\paragraph{(a) FP64 vs.\ FP32 precision at finest level.}
At 216k ($\rho=0.5$, $p=3.0$), FP64-GMG requires 23 PCG iterations in
$0.193\pm0.004$\,s, whereas FP32-GMG requires the same 23 iterations in
$0.168\pm0.014$\,s.  The direction is unchanged: FP32 smoothing carries no
iteration penalty on this moderate-condition-number case and is modestly faster
in wall time.

\paragraph{(b) Precision-depth sweep.}
Using additional FP32 coarse levels preserves convergence---PCG residual
$<10^{-6}$ in all four configurations---while the solve time varies only
modestly, from $0.161\pm0.008$\,s to $0.177\pm0.016$\,s.  The two-level point
has the smallest mean wall time, but the differences are within trial
variability, so the overall sensitivity to depth is low.

\paragraph{(c) V-cycle vs.\ W-cycle.}
W-cycle reduces outer PCG iterations from 23 to 17 (26\% fewer) at
$0.258\pm0.009$\,s total, versus V-cycle at $0.209\pm0.013$\,s for 23
iterations.  Net wall time therefore still favors V-cycle; W-cycle is
not recommended as the default.

\paragraph{(d) Chebyshev vs.\ Jacobi smoother.}
At 216k, Jacobi smoothing achieves 16 PCG iterations in
\textbf{$0.142\pm0.006$\,s}, outperforming degree-2 Chebyshev at 23
iterations in $0.198\pm0.007$\,s.
This counter-intuitive result reflects the low-contrast uniform-density
test case: on this moderate-condition-number benchmark, Jacobi's lighter
per-step cost dominates.  The data therefore support Jacobi as
the faster choice on the reported uniform 216k case.  The manuscript does not
claim that the current Chebyshev default is universally faster; it is retained
as a provisional implementation default rather than as a speed optimum.

The added high-contrast smoother screen in Table~\ref{tab:e6_high_contrast}
indicates that this default policy is not settled.  At 64k with $V_f=0.5$,
Chebyshev degree~2 has the shortest solve phase, while Jacobi degree~4 ties
its iteration count with a slightly longer solve phase.  These are solve-phase
comparisons, not setup-plus-solve totals.  At 216k with $V_f=0.5$, Jacobi
degree~2 is clearly faster in the solve phase than Chebyshev degree~2.  For the
harder $V_f=0.2$ states all four tested smoothers hit the 500-iteration cap;
degree-4 Chebyshev gives the smallest final true residual in both capped cases,
but it does not rescue convergence.  The evidence therefore supports reporting
Chebyshev as a robustness-oriented option, not as a universally superior
default.

\begin{table}[H]
\centering
\caption{High-contrast smoother-only screen using FP64-GMG with FGMRES,
  tolerance $10^{-6}$, and a 500-iteration cap.  Entries show
  iterations/solve-phase time for converged cases and final FP64 true residual
  for capped cases.  Table headings use ``Cheb.'' for Chebyshev and ``Jac.''
  for Jacobi; the number is the smoother degree.  Setup times are excluded from
  this diagnostic table.  All four high-contrast states use the fixed Bernoulli
  seed 42 and synthetic density floor $\rho_{\mathrm{floor,test}}=10^{-2}$.}
\label{tab:e6_high_contrast}
\footnotesize
\begin{tabularx}{\linewidth}{llcccc}
\toprule
Mesh & State & Cheb. 2 & Cheb. 4 & Jac. 2 & Jac. 4\\
\midrule
64k & $V_f=0.5,p=3$ & 29 / 0.500s & 49 / 1.096s & 35 / 0.568s & 29 / 0.531s\\
64k & $V_f=0.2,p=3$ & cap, $3.3{\times}10^{-2}$ & cap, $6.6{\times}10^{-3}$ & cap, $2.3{\times}10^{-2}$ & cap, $1.6{\times}10^{-2}$\\
216k & $V_f=0.5,p=3$ & 329 / 6.519s & cap, $2.1{\times}10^{-4}$ & 236 / 2.840s & cap, $5.2{\times}10^{-4}$\\
216k & $V_f=0.2,p=3$ & cap, $3.9{\times}10^{-4}$ & cap, $3.4{\times}10^{-5}$ & cap, $2.1{\times}10^{-4}$ & cap, $9.4{\times}10^{-5}$\\
\bottomrule
\end{tabularx}
\end{table}
\noindent
In this diagnostic, Chebyshev degree~2 gives the shortest solve on the
converged 64k mid-volume case, Jacobi degree~2 gives the shortest solve on the
converged 216k mid-volume case, and Chebyshev degree~4 gives the lowest final
residual on both capped near-void cases.

\begin{figure}[!htbp]
\centering
\includegraphics[width=\linewidth]{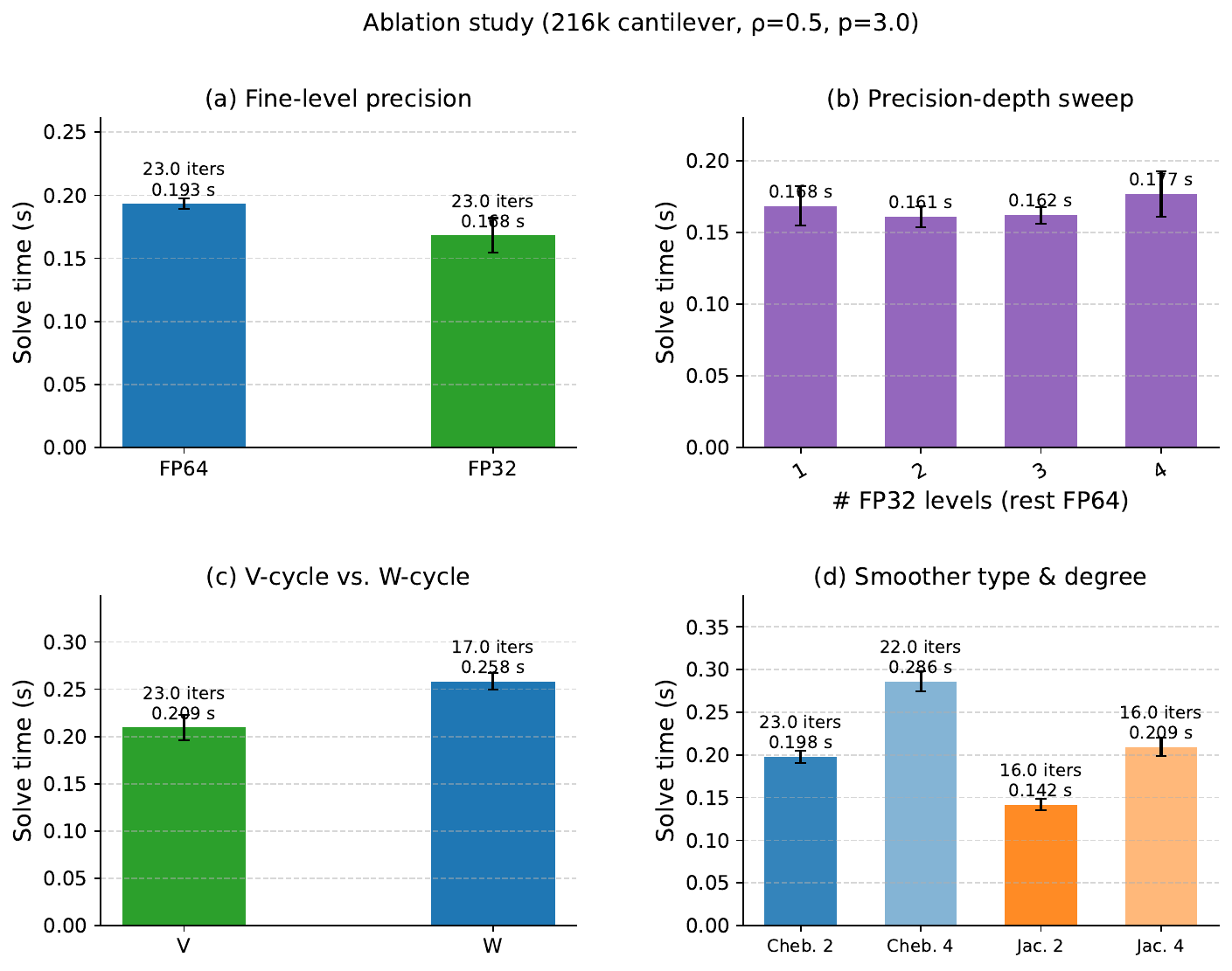}
\caption{Ablation study at 216k cantilever ($\rho=0.5$, $p=3.0$).
  (a)~Fine-level precision: FP32 matches FP64 in iterations and is modestly
  faster in wall time.
  (b)~FP32 depth sweep: additional FP32 levels change solve time only
  modestly.
  (c)~V-cycle beats W-cycle in wall time despite fewer W-cycle iterations.
  (d)~Jacobi degree~2 is fastest at uniform density; Table~\ref{tab:e6_high_contrast} reports
  the added high-contrast smoother-only screen.}
\label{fig:e6}
\end{figure}

\begin{figure}[!htbp]
\centering
\includegraphics[width=\linewidth]{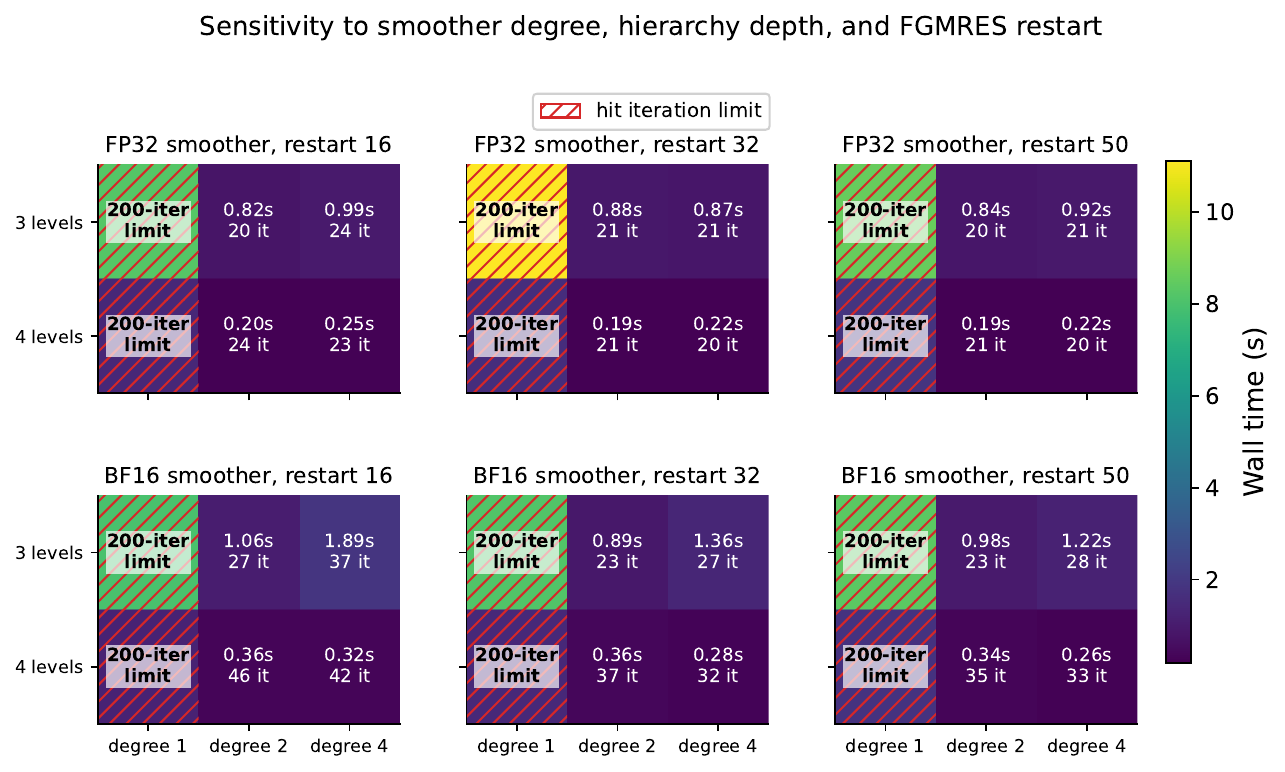}
\caption{Sensitivity of convergence and wall time to smoother degree, hierarchy depth, and FGMRES restart on the representative 216k case.
  Top row: FP32. Bottom row: BF16.  Each cell reports wall time with outer
  iterations annotated; hatched cells mark runs that hit the iteration limit.
  Degree-1 runs fail at both 3 and 4 levels for both FP32 and BF16.
  Degree-2 smoothing with the full 4-level hierarchy is near-best for FP32;
  in the BF16 screen, degree-4 with the full 4-level hierarchy is fastest.}
\label{fig:e6b}
\end{figure}

\subsection{Large-Scale Scaling to 1M Elements}\label{ssec:exp_e7}

Table~\ref{tab:e7} reports wall time and setup-time incremental VRAM delta for
single linear solves at 125k--1M elements.
Setup time ($t_{\mathrm{setup}}$) includes Galerkin coarse-operator
assembly and Chebyshev spectral-radius estimation; solve time
($t_{\mathrm{solve}}$) is the FGMRES wall clock.  The reported VRAM quantity is
the hierarchy-allocation delta measured immediately around setup
($\mathrm{VRAM}_{\mathrm{after\,setup}}-\mathrm{VRAM}_{\mathrm{before\,setup}}$),
not a peak setup-plus-solve memory trace.
The 1M-element problem (3.09M free DOFs) is solved in
$1.50\pm0.58$\,s in 18 FGMRES iterations, with a setup-time
hierarchy-allocation delta of 8.66\,GiB---about 36\% of
the RTX~4090's VRAM budget.  Larger meshes are outside the reported scaling
range of the present study.
The 1M solve-time coefficient of variation is therefore large; we report the
untrimmed timing as a noise-sensitive feasibility point rather than as a stable
throughput estimate or a strong-scaling curve.
This 1M point is a uniform-modulus solve with $E_e=0.5$; heterogeneous-field
behavior at this scale is not measured directly here and is bracketed only by
the 512k E1 stress tests.
The 512k solve is faster than the 125k solve because the outer iteration count
drops from 24 to 13 in this uniform-modulus sequence; the E7 rows should be
read as independent fixed-size timing points rather than a monotone
per-element scaling law.  The relatively large 125k setup standard deviation
and 1M setup and solve standard deviations are retained without outlier
trimming and include end-to-end hierarchy construction, GPU allocation/teardown,
and host scheduling noise.
For this reported 1M solve we use FGMRES because the active Chebyshev
preconditioner is treated as non-symmetric in the implementation; earlier symmetric
comparison paths use PCG where appropriate.

\begin{figure}[!tbp]
\centering
\includegraphics[width=0.75\linewidth]{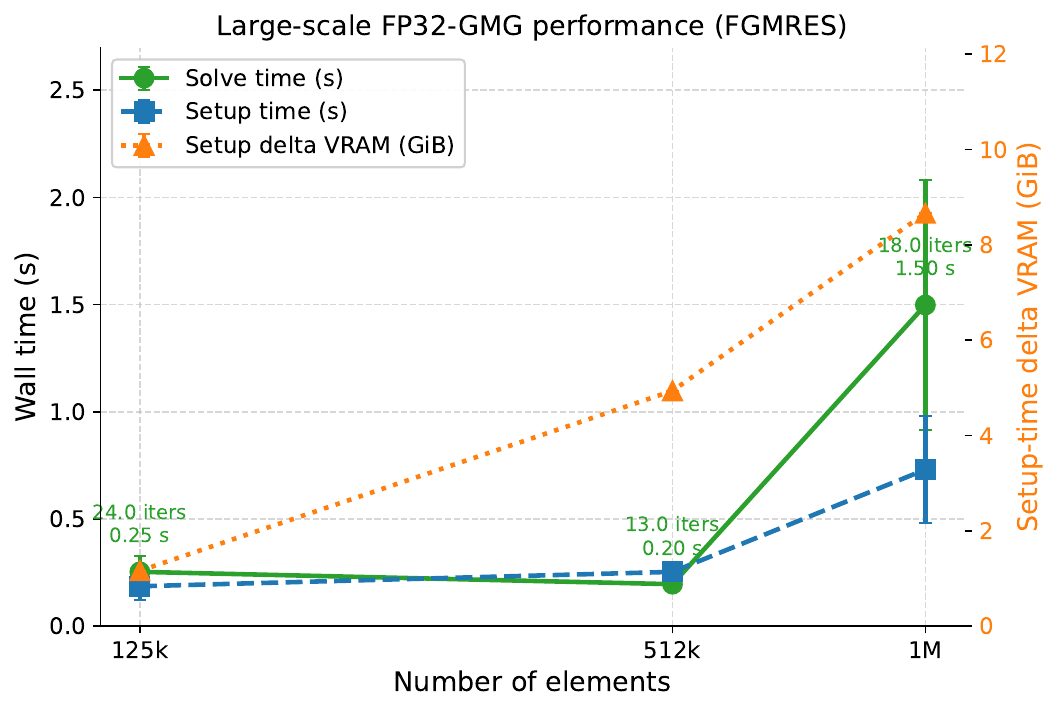}
\caption{Large-scale FP32-GMG performance.
  Left axis: setup time (dashed) and solve time (solid).
  Right axis: setup-time incremental VRAM delta measured during hierarchy
  construction.
  Iteration counts annotated on solve-time markers.}
\label{fig:e7}
\end{figure}

\begin{table}[H]
\centering
\caption{Scaling to large meshes (FP32-GMG, FGMRES outer solver,
  uniform modulus $E_e=0.5$).  Reported VRAM is the setup-time
  hierarchy-allocation delta.}
\label{tab:e7}
\begin{tabular}{lrrrrrr}
\toprule
Label & $n_{\mathrm{elem}}$ & $n_{\mathrm{free}}$ & $t_{\mathrm{setup}}$ (s) & $t_{\mathrm{solve}}$ (s) & Iters & Setup $\Delta$VRAM (MB)\\
\midrule
125k & 125,000   &   397,800 & $0.186\pm0.063$ & $0.253\pm0.073$ & 24 & 1190\\
512k & 512,000   & 1,594,080 & $0.253\pm0.004$ & $0.196\pm0.005$ & 13 & 5056\\
1M   & 1,000,000 & 3,090,600 & $0.730\pm0.251$ & $1.498\pm0.585$ & 18 & 8872\\
\bottomrule
\end{tabular}
\end{table}

\subsection{Narrow Post-Assembly PyAMG Reference}\label{ssec:exp_e8}

\begin{figure}[!tbp]
\centering
\includegraphics[width=0.60\linewidth]{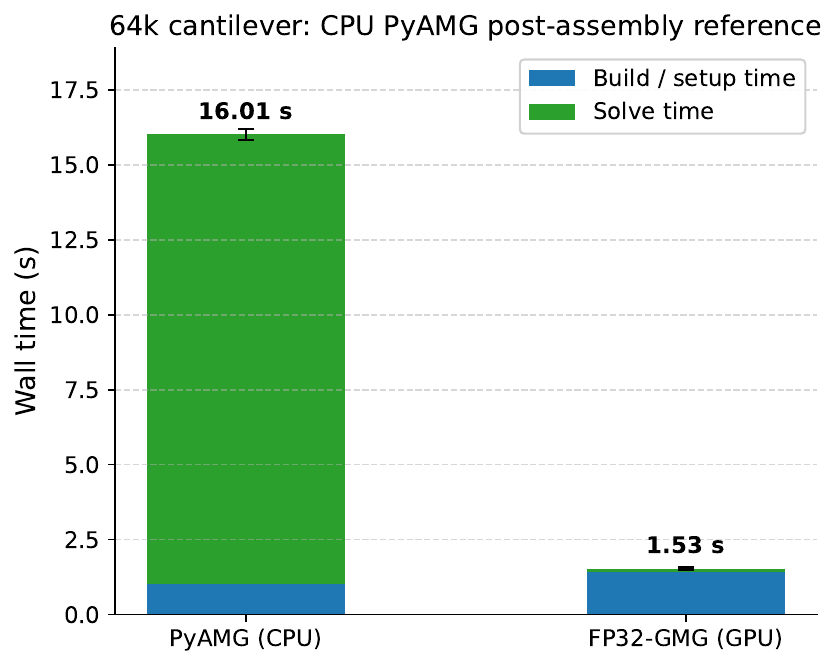}
\caption{64k cantilever build+solve after CPU CSR assembly: narrow
  CPU-after-assembly reference comparing PyAMG (CPU) with FP32-GMG (GPU).
  Stacked bars show cold-start hierarchy/setup build (blue) and solve
  (green) separately; the PyAMG timings begin after the free-free CSR matrix
  has already been assembled on the CPU.}
\label{fig:e8}
\end{figure}

\noindent\emph{Scope of this comparison.}
Experiment~E8 deliberately uses CPU PyAMG as a narrow assembled-AMG reference.
The timed PyAMG path begins only after the free-free CSR stiffness matrix has
already been assembled on the CPU, so E8 should be read as
\emph{assembled-operator build+solve after assembly}, not as a full end-to-end
assembled workflow.  That scope keeps the reported comparison on the hierarchy
construction and solve phases rather than conflating them with a separate CSR
assembly benchmark.  A follow-up comparison against an externally supplied
assembled GPU-AMG operator remains the natural next step discussed in
\S\ref{sec:discussion}.

PyAMG smoothed-aggregation AMG on the CPU requires
\textbf{$1.00\pm0.16$\,s} to build the hierarchy and
\textbf{$15.01\pm0.12$\,s} to solve a 64k cantilever linear system at uniform
$\rho=0.5$, $p=3.0$.
Running the identical system on the GPU with our FP32-GMG hierarchy
yields a measured cold-start setup and solve time of $1.43\pm0.03$\,s and
$0.104\pm0.011$\,s, respectively.  On this single narrow 64k CPU-vs-GPU
build+solve-after-assembly reference, the measured post-assembly build+solve
ratio is about $10\times$; we do not treat E8 as a general ranking of
assembled-AMG backends because no assembled GPU-AMG comparator is included.
The accompanying figure stacks the two solvers' \emph{build} and
\emph{solve} components side by side, so the comparison is based on
measured wall times rather than a synthesised estimate.  On the GMG side, the
reported setup bar is a cold-start hierarchy-initialization timing that
reconstructs the FP32 GMG object on each trial from the precomputed free-set
and element-connectivity tables.
The GPU-GMG path avoids a sparse-CSR assembly phase entirely; that conceptual
assembly step still exists for PyAMG, but it is \emph{not} included in the
reported E8 timings.

\subsection{Energy Efficiency}\label{ssec:exp_e9}

Energy is measured via NVML (\texttt{pynvml}) power sampling on the
RTX~4090 during the representative 216k cantilever solve
($120\times60\times30$, uniform $\rho=0.5$, $p=3.0$).  Total joules are
estimated by integrating timestamped NVML power samples over the measured
solve interval.  The current implementation samples NVML power every 50\,ms on a
background thread, inserts boundary samples at solve start and stop, and does
not subtract an explicit idle-power baseline, so the reported joules should be
read as an approximate solve-window energy proxy.  For each solver, we report
solve duration, average sampled power, and integrated solve-window energy.  On
the representative solve, FP64-GMG consumes
$39.08\pm6.20$\,J at $164.7\pm25.5$\,W average and FP32-GMG consumes
$28.37\pm8.23$\,J at $145.2\pm41.2$\,W average; we do not report a bare-metal
efficiency ratio because the CPU
PyAMG baseline is not directly comparable in thermal envelope.
These solve-window values include idle draw because no idle-power baseline is
subtracted.  The full benchmark-suite energy and setup overhead were not
measured.

\subsection{Robustness Across Density Configurations}\label{ssec:exp_e10}

\begin{figure}[!tbp]
\centering
\includegraphics[width=\linewidth]{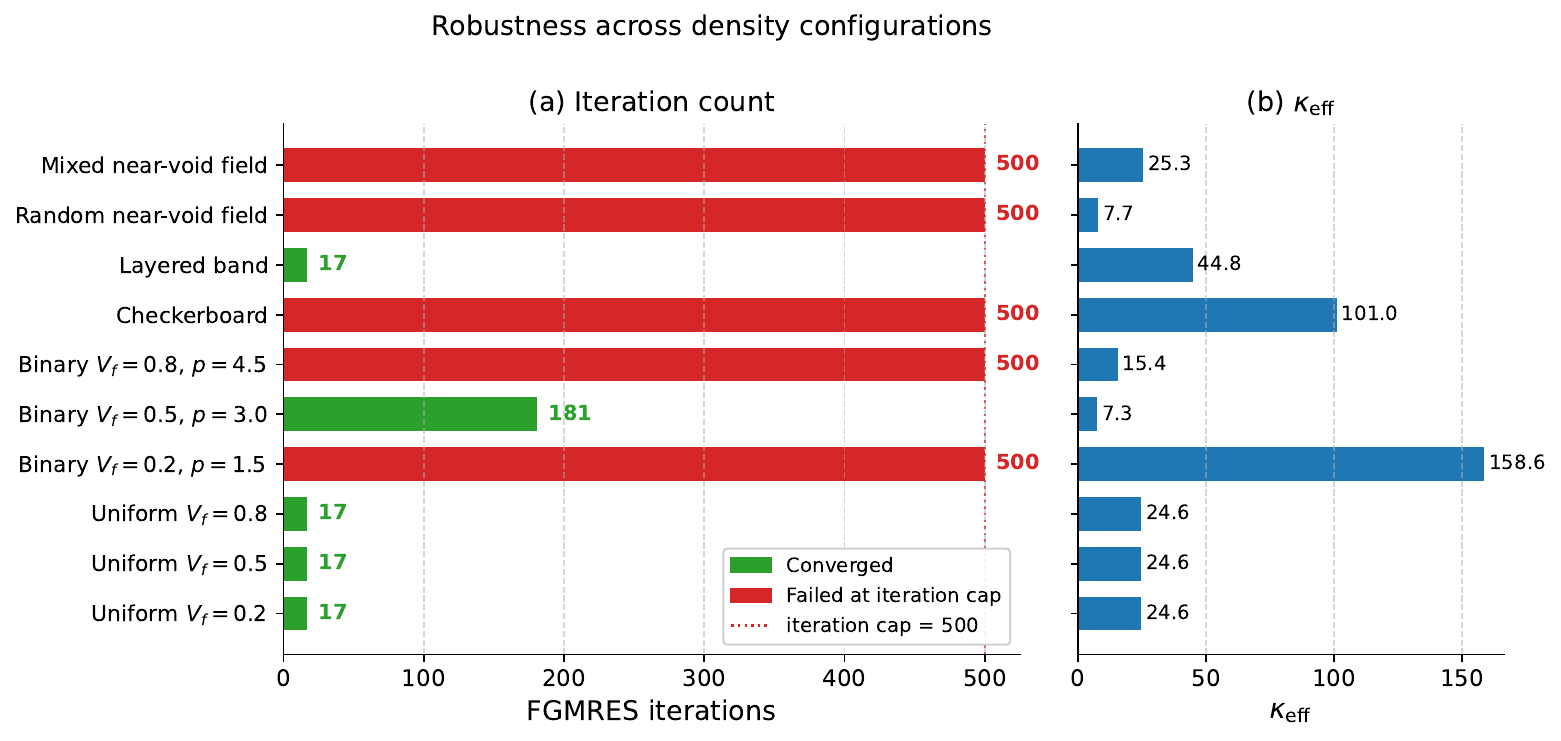}
\caption{Robustness across ten representative fixed-seed or deterministic
  stress configurations on the 64k cantilever with the FP64 hierarchy.
  (a) FGMRES iterations (green = converged, red = failed at the iteration cap).
  (b) Effective condition number $\kappa_{\mathrm{eff}}$ for the fine-system
  linear operator under the full hierarchy.
  The failures are concentrated in highly heterogeneous or near-singular
  fields rather than in the moderate-contrast uniform baselines.}
\label{fig:e10}
\end{figure}

\begin{figure}[!tbp]
\centering
\includegraphics[width=0.88\linewidth]{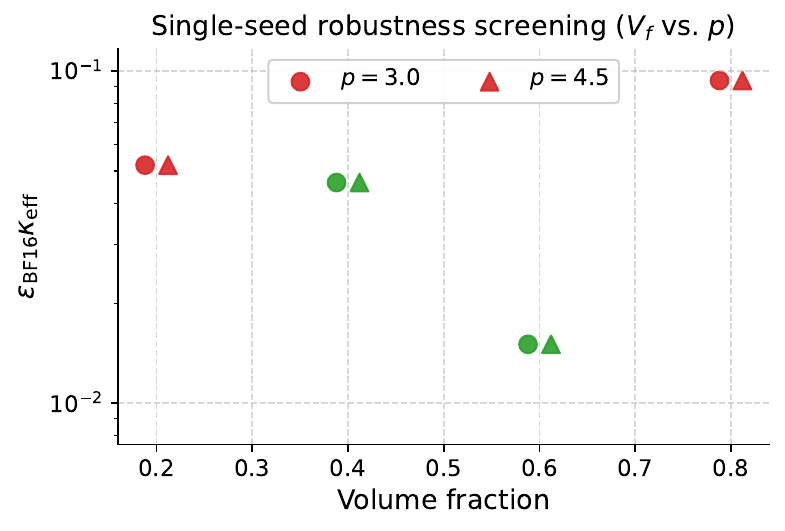}
\caption{Single-seed robustness screening map over volume fraction and
  penalization on the 64k cantilever with the FP64 hierarchy.  Circles denote
  $p=3.0$, triangles denote $p=4.5$, green markers converged, and red markers
  hit the iteration cap.  This sweep uses the fixed Bernoulli seed 23 together
  with restart 50 and a 300-iteration cap as a bounded-cost screening pass,
  whereas Table~\ref{tab:e10} uses a 500-iteration cap.  For the tested low-floor labels
  $\rho_{\mathrm{floor,test}}\in\{10^{-12},10^{-9},10^{-6}\}$, the SIMP map with
  $E_{\min}=10^{-9}$ collapses them to essentially the same effective modulus
  floor, so the figure shows the corresponding collapsed representative points
  rather than three visually redundant panels.  The vertical axis uses
  $\eps_{\mathrm{BF16}}\kappa_{\mathrm{eff}}$ to connect the screening map to
  the mixed-precision spectral-proxy discussion.  The visual markers therefore
  aggregate the three tested construction floors, not three independent
  frequency estimates.}
\label{fig:e10b}
\end{figure}

Table~\ref{tab:e10} summarises robustness outcomes on the
64k cantilever across ten representative stress configurations, ranging
from benign uniform states to binary-contrast, layered, and near-singular
pathologies.  The stochastic cases use the fixed seeds listed in
Table~\ref{tab:repro_defaults}.  FGMRES is
used throughout with restart 50 and a 500-iteration cap; results are
taken directly from the robustness summary record.
Here $\rho_{\min}$ and $\rho_{\mathrm{floor,test}}$ denote density-field
construction floors for the synthetic stress cases, while $E_{\min}=10^{-9}$
is the SIMP stiffness floor used in the modulus map.

\begin{table}[!tbp]
\centering
\small
\caption{Robustness summary (FGMRES, restart 50, 500-iteration cap).
  ``Pass'' = converged to relative residual $<10^{-6}$;
  ``Fail'' = iteration cap exceeded.}
\label{tab:e10}
\begin{tabular}{lllr}
\toprule
Configuration & Result & Iters & $\kappa_{\mathrm{eff}}$\\
\midrule
Uniform $V_f=0.2$                                    & Pass & 17  & 24.6\\
Uniform $V_f=0.5$                                    & Pass & 17  & 24.6\\
Uniform $V_f=0.8$                                    & Pass & 17  & 24.6\\
Binary $V_f=0.2$, $p=1.5$                            & Fail & 500 & 158.6\\
Binary $V_f=0.5$, $p=3.0$                            & Pass & 181 & 7.3\\
Binary $V_f=0.8$, $p=4.5$                            & Fail & 500 & 15.4\\
Checkerboard density                                 & Fail & 500 & 101.0\\
Layered half-solid / half-void band                  & Pass & 17  & 44.8\\
Random field with $\rho_{\mathrm{floor,test}}=10^{-12}$ & Fail & 500 & 7.7\\
Mixed near-void field                                & Fail & 500 & 25.3\\
\bottomrule
\end{tabular}
\end{table}

\noindent
In these fixed-seed screening cases, a clearer pattern emerges.  Moderate
uniform states converge in 17 iterations.  Binary
high-contrast fields remain solvable at the mid-volume point
$V_f=0.5$, $p=3.0$, but the lower- and higher-volume binary cases already hit
the 500-iteration cap.  Checkerboard,
very-low-floor, and mixed near-void states still fail, suggesting that the
failures are consistent with near-singular local structure and are not
explained by BF16 precision alone.
Because the tested $\rho_{\mathrm{floor,test}}$ labels collapse to essentially the same
effective modulus floor under the present SIMP map, Figure~\ref{fig:e10b}
should be read as a collapsed volume-fraction/penalization screening map
rather than as an informative floor sweep.  It correlates only
loosely with $\eps_{\mathrm{BF16}}\kappa_{\mathrm{eff}}$: the Lanczos-based
$\kappa_{\mathrm{eff}}$ probe remains a useful spectral diagnostic, but it
does not by itself detect the near-null modes
that dominate these failure cases.  Because Figure~\ref{fig:e10b} uses the
lighter 300-iteration screening budget, it should be read as a
single-seed screening map
rather than as a one-to-one duplicate of the 500-iteration table.

This five-of-ten failure rate is the central open problem identified by the
present work, and we report it explicitly rather than restricting the benchmark
suite to the compliant subset.  Spectrum-aware Galerkin preconditioning
succeeds on the moderate-contrast SIMP cases represented in the reported
benchmark suite but degrades on the binary high-contrast and near-void regimes that
arise late in continuation.  This pattern is consistent with a loss of
coarse-space adequacy under high contrast, although the present study does
not isolate that mechanism directly.
Closing this gap --- through pathology-aware smoothing, contrast-adaptive
coarsening, or per-level failure detectors that detect near-null modes
that $\kappa_{\mathrm{eff}}$ alone misses --- is the natural follow-up to the
present hierarchy and is restated as an explicit limitation in
\S\ref{sec:discussion}.

\section{Qualitative Topology Gallery}\label{sec:gallery}

Figure~\ref{fig:gallery} provides qualitative geometry context for the benchmark
family.  All structures are rendered as isosurfaces
($\hat{\rho}_e>0.5$ shown as solid) using the marching-cubes
algorithm~\citep{lorensen1987marching}.  These panels are not quantitative
evidence for solver robustness; that evidence is reported in
Figures~\ref{fig:e10}--\ref{fig:e10b} and Table~\ref{tab:e10}.  Appendix
Table~\ref{tab:gallery_sources} maps each reader-facing panel label to the
stored density field used for traceability.  For each retained auxiliary run,
Figure~\ref{fig:gallery} renders the retained checkpoint satisfying the local
run-validity checks, selected before rendering and not by visual appearance,
solver residual, or compliance value.  The low-volume 1M cantilever panel is a
retained qualitative snapshot rather than a selected best-valid state.  The run
set was selected to span benchmark families and mesh sizes, not as quantitative
solver evidence.  Appendix~\ref{app:artifact} documents the artifact trace for
these qualitative panels.

\begin{figure}[!htbp]
\centering
\includegraphics[width=\linewidth]{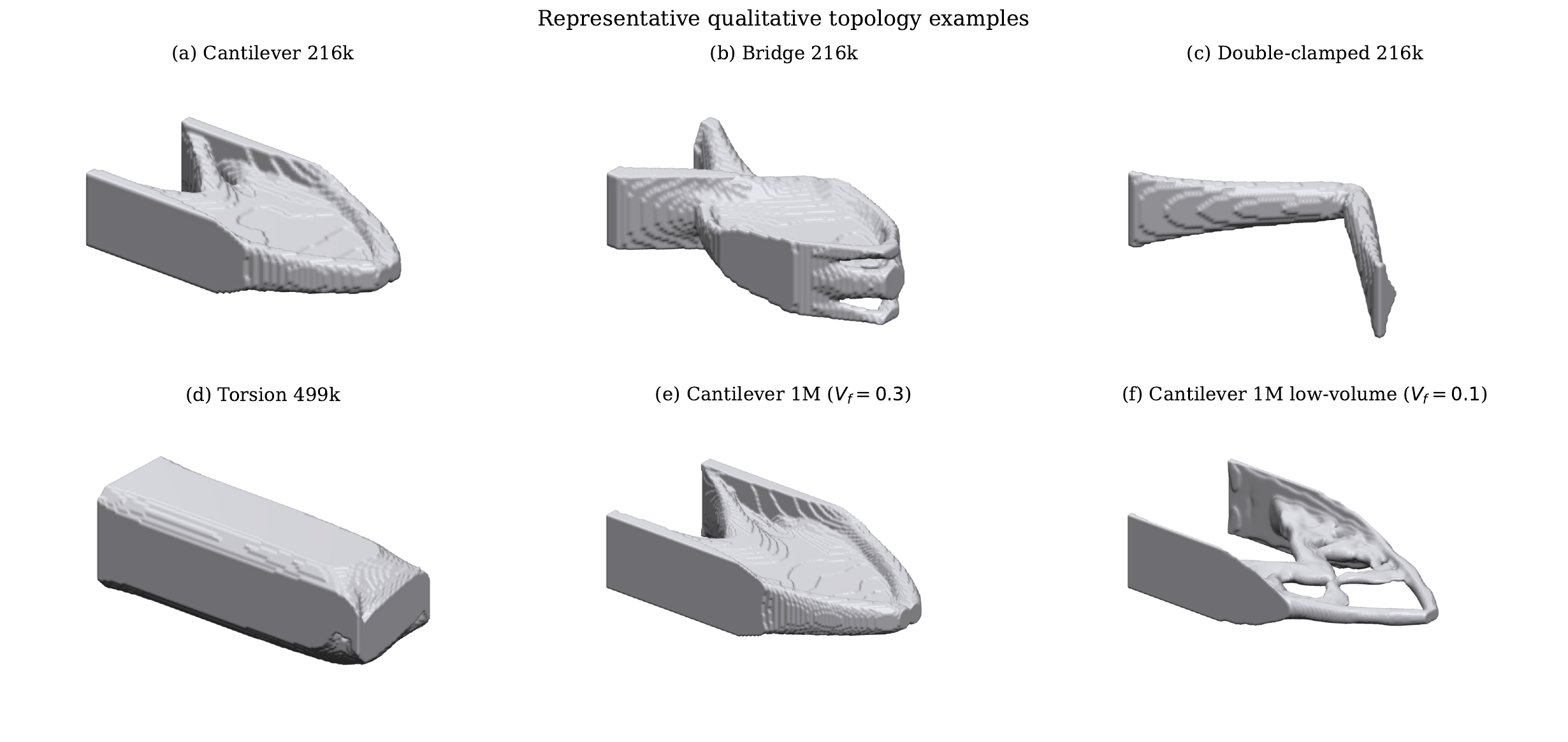}
\caption{Qualitative auxiliary density fields rendered as structure examples.
  \emph{Top row (left to right):} (a) Cantilever 216k
  (120$\times$60$\times$30, $V_f=0.3$); (b) bridge 216k
  (120$\times$60$\times$30, $V_f=0.3$); (c) double-clamped beam 216k
  (120$\times$60$\times$30, $V_f=0.1$).
  \emph{Bottom row:} (d) torsion 499k (165$\times$55$\times$55,
  $V_f=0.25$); (e) cantilever 1M (200$\times$100$\times$50, $V_f=0.3$);
  (f) cantilever 1M low-volume qualitative snapshot
  (200$\times$100$\times$50, $V_f=0.1$).
  All renders use an isosurface threshold $\hat{\rho}=0.5$ and light
  Taubin smoothing for presentation (10 iterations, pass band 0.1).
  The displayed fields are stored
  best-valid or retained qualitative snapshots from the selected auxiliary
  runs; ``best-valid'' means the retained checkpoint satisfying the local
  run-validity checks, not a panel chosen by visual appearance or solver
  residual.  The run set was selected to span benchmark families and mesh
  sizes, not as quantitative solver evidence.}
\label{fig:gallery}
\end{figure}

\FloatBarrier

\section{Discussion}\label{sec:discussion}

\subsection{Mechanistic Analysis}

The central empirical point is that the hierarchy helps primarily by changing
the \emph{solver regime}, not by exposing a large raw-throughput advantage for
BF16.  Experiments~E1 and E2 show that FP32-GMG lowers wall time while reducing
outer iterations relative to the capped Jacobi-PCG path, consistent with a more
effective preconditioned spectrum.  Experiment~E4 then shows that the BF16 and FP32 fused
fine kernels are effectively tied on the proxy benchmark
(3\,046 vs.\ 3\,242\,GFLOP/s).  That tie is not accidental.  The local operator
is only $24\times24$, so the WMMA path must pad to $32\times32$, cast the
element tile, preserve FP32 accumulation, and scatter through the same
memory-bound gather/scatter structure as the FP32 path.  In other words, the
tensor-core arithmetic is not the bottleneck by itself; the full smoother still
pays for register pressure, padding overhead, and limited arithmetic intensity.
The consequence is exactly what the timing and direct-validation studies
report: BF16 can track FP64 on the states where the hierarchy converges, but it
is not yet the fastest or most robust end-to-end variant.

This also explains the three negative ablation results.  First, W-cycle reduces
outer iterations but loses in wall time because it doubles the coarse-grid work
before the finest-level bottleneck has been relieved; the saved Krylov steps are
not numerous enough to amortise the second descent and ascent.  Second,
Jacobi degree~2 beats Chebyshev degree~2 at uniform density because the test state is
already spectrally mild.  Once $\kappa_{\mathrm{eff}}$ is modest, the extra
matvecs inside Chebyshev are overhead rather than protection, so the current
Chebyshev default should be read as a provisional robustness-oriented
engineering choice rather than as a demonstrated speed optimum on every state.
The high-contrast smoother screen sharpens this point: Chebyshev degree~4 lowers
the final residual on the two capped near-void states, but Jacobi degree~2 is much
faster on the 216k mid-volume state, so no tested smoother is uniformly best.
Third, BF16 does not improve end-to-end wall time because the present
$24$-DOF trilinear hex operator is too
small and too bandwidth-dominated for tensor-core arithmetic alone to dominate
wall time.  These are negative results, but they are informative: they localise
the value of the current hierarchy to moderate-regime convergence, memory
feasibility, and explicit failure-mode diagnosis rather than to indiscriminate
acceleration of every regime.

\subsection{Failure Modes and Spectral Drift}

The main limitation in the results is not a single BF16 outlier in the
spectral-proxy screen; it is the heterogeneous-field behavior of the hierarchy itself.
In that direct validation, BF16 and FP64 converge on the same number of cases
(7/18), while many screened-in states still reach the 500-iteration cap.  In E1, the same kind of
stagnation appears in the \emph{FP64} hierarchy at 512k.  These facts indicate
that the present structured hierarchy is not yet sufficiently robust to
maintain mesh-independent behavior under extreme binary-contrast fields at
larger problem sizes.  Put differently, the BF16 precision-sensitivity problem and the
GMG approximation-quality problem are related but distinct: the proxy can flag
a potentially dangerous mixed-precision spectrum, but it does not detect all
coarse-space or smoother failures.

The single screened-out point is a near-void 216k state
($V_f=0.2$, $p=1.5$), where the screen spikes to
$\eps_{\mathrm{BF16}}\kappa_{\mathrm{eff}}\approx 12.59$ even though both FP64
and BF16 converge.  Conversely, eleven screened-in BF16 cases do not converge
within the 500-iteration cap.  The broader robustness screens show the same
lesson: late high-contrast binary states can fail even when this scalar proxy
is not extreme.  The observed failure pattern is consistent with SIMP-induced
contrast reducing the adequacy of the current coarse-space / smoother balance,
although the present study does not isolate that mechanism directly.
These sampled cases suggest that the hierarchy behaves best in the middle of
the continuation path, where
contrast is large enough to require preconditioning but not so extreme that the
coarse space ceases to represent the fine correction well.  A full
$\kappa_{\mathrm{eff}}$ trajectory over all 30 SIMP steps is therefore the next
diagnostic to add, because isolated $(V_f,p)$ samples already indicate that the
relevant drift is tied to contrast evolution rather than to mesh size alone.

The robustness table reinforces the same interpretation.
Checkerboard and mixed near-void states fail despite modest raw
$\kappa_{\mathrm{eff}}$ estimates, which means the Lanczos-based screening
probe does not capture the near-null directions responsible for those
pathologies.
That is a useful caution for future work: $\kappa_{\mathrm{eff}}$ is an
interpretable spectrum diagnostic, but it is not an admissibility certificate or
a complete pathology detector.  Per-level smoothing factors and
coarse-grid correction histories are the measurements needed to separate a
coarse-space defect from a fine-smoother defect more cleanly.

\subsection{AMG Comparison and Engineering Implications}

Experiment~E8 should be read carefully.  The PyAMG comparison does not claim
to exhaust the external AMG baseline space; it shows
something narrower and still important.  The reported PyAMG timings measure
hierarchy construction and iterative solve only after the free-free CSR matrix
already exists, so the comparison is a post-assembly assembled-reference case.
With the cold-start GMG timing measured on a freshly reconstructed object per
trial, the setup advantage disappears: on this 64k problem the
FP32-GMG setup stage is slower than the PyAMG build stage
($1.43\pm0.03$\,s versus $1.00\pm0.16$\,s).  The remaining overall speedup
therefore comes from the solve phase, where the matrix-free GPU path is
still much faster on the reported reference case.  That conclusion remains
specific to this single post-assembly 64k CPU-vs-GPU comparison and is not a
general ranking of all assembled-AMG backends.
Other assembled GPU-AMG libraries remain plausible alternatives when the mesh
is irregular, when many right-hand sides amortise setup, or when an assembled
sparse operator already exists for other reasons.  The present GMG path is
favorable in the structured, matrix-free, single-right-hand-side regime tested
here under a tight memory budget.

From an engineering perspective, the main practical result is that a consumer
GPU can reach the $10^6$-element regime for the reported uniform-modulus
linear solve without reverting to a cluster-style sparse assembly workflow.  The
reported 1M-element case corresponds to 3.09 million free DOFs, a
$1.50\pm0.58$\,s solve time, and an 8.66\,GiB hierarchy-allocation delta
during setup, not a measured peak-memory trace, on an
RTX~4090.  Likewise, on the auxiliary 30-step fixed-penalty
schedule, the cantilever-216k run drops from 73.4\,s to 32.8\,s, but that
change includes repeated capped baseline solves rather than a matched-endpoint
optimization win.  The practical meaning is not that a single
GPU supersedes extreme-scale cluster workflows such as \citet{aage2017giga};
rather, the 1M-element uniform-modulus result suggests that this linear-solve
regime can be reached on a consumer GPU with a suitable hierarchy.  Solver
design appears to be the primary constraint for similar moderate-contrast
linear solves, while heterogeneous-field behavior at this scale remains
unmeasured.

\subsection{Limitations}

The present hierarchy is specialised in several ways.  First, it assumes a
structured Q1 hexahedral mesh with geometric $2{:}1$ coarsening.  Extending it
to unstructured meshes would require either algebraic transfer operators or a
mesh hierarchy capable of preserving the same boundary-mask semantics.  Second,
the analysis and experiments are restricted to linear elasticity with isotropic
SIMP interpolation.  Nonlinear constitutive laws or multi-material
interpolations would require re-evaluating the spectral proxy around the
resulting tangent operators, not merely rerunning the same code.

Third, the implementation is single-GPU only.  Multi-GPU execution would need a
distributed representation of the transfer operators, overlap exchange on the
fine matrix-free action, and a clear policy for where the first assembled coarse
level lives.  Fourth, the fine kernel is tuned to the $24$-DOF trilinear hex
operator.  Higher-order elements are actually attractive for tensor cores
because their local dense blocks are larger, but they would require a new tile
layout, different transfer operators, and a fresh cost model.  Finally, the
reproduction bundle includes repeated-run timing statistics, residual histories,
SIMP-step trajectories, direct BF16 validation on the heterogeneous probe,
high-contrast smoother checks, and a roofline proxy, but it still lacks three
measurements that would be required for a stronger HPC comparison: multi-seed
BF16 validation on real filtered continuation states, a broader assembled
baseline set such as AmgX plus PETSc GAMG or hypre BoomerAMG with and without
assembly cost, and direct Nsight/CUPTI kernel counters.  It also lacks a
successful high-contrast 1M-element SIMP-like state; the present 1M result is
explicitly uniform-modulus only.  We also report solve-window energy only; total
benchmark energy, including setup and capped or failed screening runs, was not
instrumented.  Those remaining
measurements matter because they would separate algorithmic gains from
backend-specific implementation effects more cleanly than the present study can.

\FloatBarrier

\section{Conclusion}\label{sec:conclusion}

The main result of this paper is not that BF16 is universally faster; it is
that an important algorithmic step for the studied matrix-free 3D SIMP setting
is a spectrum-aware hierarchy.  Once the flat
Jacobi preconditioner is replaced by a matrix-free Galerkin GMG cycle, the
fine-level work can be evaluated in BF16 on the subset of tested heterogeneous
states where the FP64 hierarchy itself converges, but without a wall-time
advantage over FP32.  The accompanying measurements also clarify the limits of
that statement: the
$\eps_{\mathrm{BF16}}\kappa_{\mathrm{eff}}<1$ proxy is useful for interpretation
but not a convergence classifier, the benefit comes primarily from
moderate-regime convergence and memory feasibility rather than from a large
tensor-core throughput win, and the hierarchy still degrades on strongly
heterogeneous 512k states and
near-singular density pathologies.

The evidence therefore supports a narrower conclusion than a generic
``BF16 accelerates topology optimization'' claim.  It shows that the
single-GPU $10^6$-element regime can be reached for the reported uniform-modulus
linear solve with a matrix-free solver stack that remains
Galerkin-consistent and analytically interpretable.  The reported speed
advantages over Jacobi-PCG should be read as advantages over a capped
stagnating reference path, while the PyAMG comparison remains a narrow 64k
post-assembly reference.  The smoother ablation likewise supports a bounded
engineering conclusion: Jacobi degree~2 is faster on mild states, while the
Chebyshev option is retained as a robustness-oriented choice rather than as a
demonstrated wall-time optimum on every state.  Future work should strengthen
the hierarchy for
pathological SIMP states, add converged external baselines including an
assembled GPU-AMG operator, and extend the transfer-operator logic beyond the
present structured Q1 hex setting.

\section*{Code and data availability}

The public GitHub repository address has been reserved at
\href{https://github.com/nbbllxx0/Mixed-Precision-GMG-SIMP}{the Mixed-Precision-GMG-SIMP GitHub repository};
the code payload is staged locally for upload there when the arXiv version of
this paper is online.  The release will be distributed under the BSD 3-Clause
license and will include a comprehensive README for reproducing the experiments
from scratch.
The Level~0
fused gather--GEMM--scatter
operator, treated here as fixed infrastructure, is released separately by
\citet{yang2026fused} as
\href{https://arxiv.org/abs/2604.18020}{arXiv:2604.18020}, with companion
repository
\href{https://github.com/nbbllxx0/Fused-Gather-GEMM-Scatter-Kernels}{github.com/nbbllxx0/Fused-Gather-GEMM-Scatter-Kernels}.
Appendix~\ref{app:artifact} gives the reproduction workflow and maps each
reported result to its provenance record.  No external dataset is used.  The
public release will include the implementation, environment specification, and
documented workflow needed to regenerate the reported measurement records
locally.  Generated measurement outputs, logs, retained exploratory outputs,
retained density arrays, and generated figures are not part of the public code
release.  The arXiv source package contains the rendered figure files needed to
compile the paper, not the raw gallery density arrays.

\appendix

\section{Auxiliary OC Schedule Diagnostic}\label{app:e3_aux}

Figure~\ref{fig:e3} and Table~\ref{tab:e3} report total wall time for an
auxiliary 30-step fixed-penalty OC schedule on three benchmarks.  This auxiliary
experiment is deliberately narrower than the full SIMP driver described in
Section~\ref{sec:method}: it uses the fixed OC update with fixed
penalization $p=3.0$ and does not replay the full
filter/projection-continuation machinery.  It also uses a different
linear-solve tolerance from the $10^{-6}$ default used in Experiments~E1, E2,
and E5--E10: both solver stacks use $10^{-5}$ in this auxiliary schedule, so
E3 should be read as a same-schedule execution diagnostic rather than as a
tighter $10^{-6}$ state-solve study or a matched-endpoint optimization
comparison.

Because the trajectories separate after repeated capped baseline solves, this
appendix reports execution cost only, not optimization endpoint quality.  Under
that diagnostic framing, FP32-GMG executes the 30 steps $2.24\times$ faster on
cantilever-216k, $3.09\times$ faster on torsion (3k elements), and
$1.38\times$ faster on the Messerschmitt--Bölkow--Blohm (MBB) beam
(1.5k elements).  The torsion and MBB cases are small reference problems, so
the reported same-schedule execution ratios should be read mainly as
solver-stagnation avoidance within the auxiliary OC loop rather than as
GPU-throughput evidence at the 64k--1M scales emphasized elsewhere in the
paper.  The meaningful parity check is the first SIMP step, where the GMG path
remains within 0.108\%, 0.001\%, and 0.062\% of the matrix-free Jacobi-PCG
baseline of \citet{yang2026fused} on the three benchmarks.

The baseline nevertheless hits the 1000-iteration cap on 27/30 cantilever
steps, 27/30 torsion steps, and 25/30 MBB steps, while the GMG path hits that
cap once on MBB.  Figure~\ref{fig:e3b} then shows why the final-step compliance
values diverge much more strongly: once the two solvers start taking different
approximate state solves, the trajectories separate and the later steps are no
longer a controlled like-for-like comparison of the same design path.  The
final-step compliance gaps retained in the data files (303.7\%, 122.4\%, and
347.6\%) therefore reflect trajectory separation plus repeated capped solves,
not a controlled same-state solver error.  The reported wall times compare work
to execute the same auxiliary 30-step schedule, not work to reach the same
final design.

\begin{figure}[!tbp]
\centering
\includegraphics[width=0.75\linewidth]{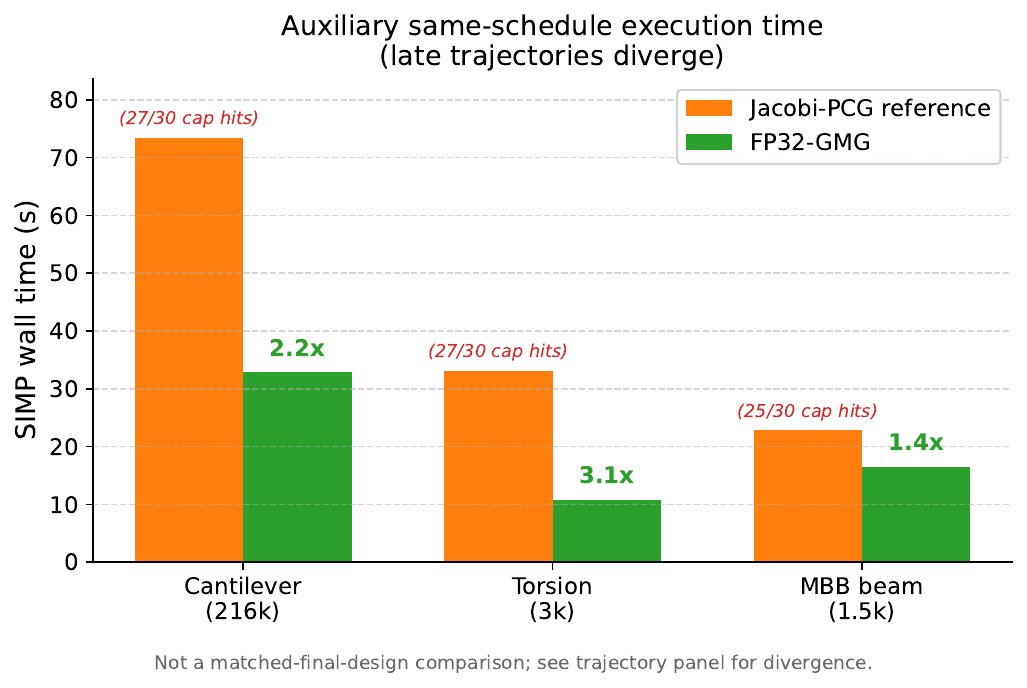}
\caption{Schedule-execution wall time on the auxiliary 30-step fixed-penalty OC
  schedule; this is \emph{not} a matched-endpoint optimization comparison.
  The Jacobi-PCG reference path follows \citet{yang2026fused}.  The
  Jacobi-PCG baseline hits the 1000-iteration cap on 27/30 cantilever steps,
  27/30 torsion steps, and 25/30 MBB steps; the GMG MBB run hits the cap once.
  Ratio labels annotate the GMG bars; final-step compliance diverges strongly,
  so Figure~\ref{fig:e3b} and Table~\ref{tab:e3} must be read together.}
\label{fig:e3}
\end{figure}

\begin{figure}[!tbp]
\centering
\includegraphics[width=\linewidth]{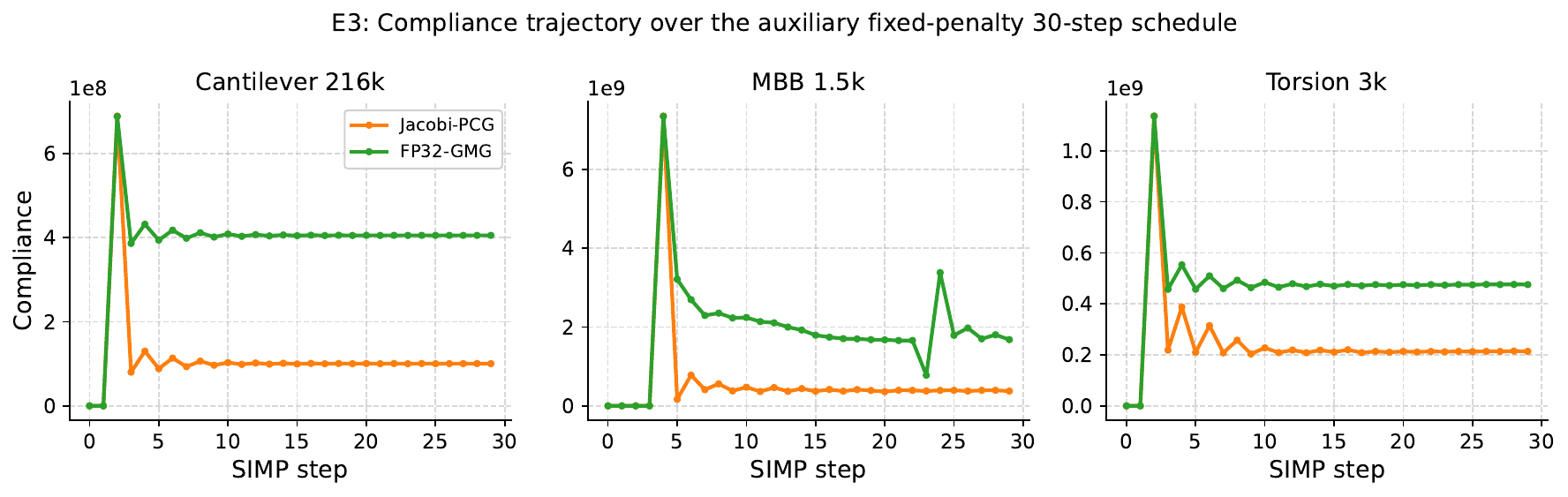}
\caption{E3: Compliance trajectories over the auxiliary 30-step fixed-penalty OC
  schedule.  The k=0 agreement is tight, but later SIMP states separate once
  the two solver stacks induce different approximate optimization paths.
  Panels use separate y-scales so the within-benchmark trajectory separation
  remains visually readable.}
\label{fig:e3b}
\end{figure}

\begin{table}[!tbp]
\centering
\caption{Schedule-execution wall time on the auxiliary 30-step fixed-penalty OC
  schedule, not wall time to reach matched final designs.  Each row is a
  single 30-step schedule measurement, so no standard deviation is estimated.
  The compliance comparison is reported at SIMP step k=0 (uniform design); the
  data files also store final-step compliance errors of 303.7\%, 122.4\%, and
  347.6\%, which are not used in the main comparison because the trajectories
  separate and many baseline steps hit the iteration cap.}
\label{tab:e3}
\small
\begin{tabularx}{0.95\linewidth}{llrrr>{\RaggedRight\arraybackslash}X}
\toprule
Benchmark & Size & Jacobi-PCG (s) & FP32-GMG (s) & Time ratio & Notes\\
\midrule
Cantilever & 216k & 73.4 & 32.8 & $2.24\times$ & 27/30 Jacobi-PCG steps hit cap\\
Torsion    & 3k & 33.1 & 10.7  & $3.09\times$  & 27/30 Jacobi-PCG steps hit cap\\
MBB beam   & 1.5k & 22.8 & 16.5 & $1.38\times$  & 25/30 Jacobi-PCG, 1/30 GMG hit cap\\
\bottomrule
\end{tabularx}
\end{table}

\section{Reproducibility}\label{app:repro}

This appendix records the settings needed to interpret and rerun the reported
experiments.

Software and hardware context.
All GPU experiments for the present implementation use the single RTX~4090
configuration stated in the experimental setup subsection.  The
recreate-from-scratch benchmark path uses Python~3.10.18, pip~24.3,
\path{cupy-cuda12x}~13.6.0, NumPy~2.2.6, SciPy~1.15.3,
Matplotlib~3.10.7, pandas~2.3.3, PyVista~0.46.3, VTK~9.6.0,
scikit-image~0.25.2, PyAMG~5.3.0, \path{pynvml}~13.0.1, and
\path{nvidia-ml-py}~13.595.45.  The reported quantitative results were
produced on an RTX~4090 host whose NVIDIA driver reports CUDA API support
13.2, with execution through the CUDA~12.x \path{cupy-cuda12x} runtime
package.  The source package includes \path{environment.yml} and
\path{RUNTIME_SNAPSHOT.md} as machine-readable reproduction notes.  The
current manuscript PDF and figures were regenerated from those benchmark
outputs during the final manuscript build; this is not a claim of byte-identical
rendering from the pinned benchmark environment alone.

Public code-release timing.
The public GitHub repository address has been reserved at the address above;
the code payload is staged locally for upload there when the arXiv version is
online.  The public release will contain the implementation, experiment
drivers, figure-generation scripts, comprehensive reproduction README,
and environment specification.  It is intended for independent reruns rather
than as a generated-output, log, density-field, or figure release.

Key interpretation rules.
All reported quantitative values are associated with the reproduction steps and
output artifact types listed in Table~\ref{tab:artifact_results}.  The public
code release is designed to regenerate these outputs locally.  Residual-history,
trajectory, and trial-level outputs are supporting rerun records rather than
separate headline results.  The
qualitative
Figure~\ref{fig:gallery} gallery is
presentation-only; it is not generated from the E1--E10 quantitative result
files and should not be read as a quantitative substitute for the measured
benchmark suite.  The retained qualitative gallery density fields are local
author-side provenance and are not part of the arXiv source package or public
code release.  The validation step records the M1--M8 checks behind
Table~\ref{tab:phase1}.  The full benchmark suite is expected to take
approximately 4--6\,h on the RTX~4090 workstation at the default statistical
setting.  Total benchmark energy, including setup, failed or capped screening
runs, and figure generation, was not measured; the energy measurements in
Section~\ref{ssec:exp_e9} are solve-window proxies only.

After that code repository is released and cloned, following the release README
to enter the experiment-driver and figure-generation directories, and creating the
supplied conda environment (\texttt{conda env create -f environment.yml};
\texttt{conda activate gmg-simp}), the workflow is as follows.
The environment name is user-selectable; the retained author-side snapshot used
a local environment named \texttt{Pytorch} with the package pins listed in
\path{environment.yml}.
\noindent Set output directories outside the manuscript bundle before rerunning.
For POSIX shells:
\begin{verbatim}
export OUTDIR=../../rerun_outputs/paper4
export FIGDIR=../../rerun_outputs/paper4_figs
export GALLERYDIR=/path/to/local/gallery_density_arrays
\end{verbatim}
For PowerShell:
\begin{verbatim}
$OUTDIR='../../rerun_outputs/paper4'
$FIGDIR='../../rerun_outputs/paper4_figs'
$GALLERYDIR='C:/path/to/local/gallery_density_arrays'
\end{verbatim}

\begin{enumerate}
\item \textbf{Pre-benchmark validation checks} (must pass before the main
  benchmark suite; run from the
  experiment-driver directory):
\begin{verbatim}
python validate_phase1.py --out results_phase1.json
\end{verbatim}
This runs the M1--M8 correctness checks tabulated in the code and provenance
appendix:
FP64 V-cycle vs.\ direct solve (M1), outer-iteration
count vs.\ mesh size (M2), matrix-free vs.\ assembled-Galerkin compliance (M3),
Chebyshev vs.\ Jacobi smoother (M4), selected SIMP sanity probes (M5),
$\kappa_{\mathrm{eff}}\!\le\!256$ probe (M6), BF16 drop-in compliance using
FGMRES restart 50 and maxiter 2000 as a compliance-only gate (M7), and
three-level FP32 hierarchy (M8).
\item \textbf{Main benchmark suite E1--E10.}
Run from \path{experiments/paper4/}, not from the manuscript directory, and
write to a separate rerun tree so the shipped manuscript bundle is not
overwritten:
\begin{verbatim}
python run_experiments_e1_e10.py --experiments all --out $OUTDIR
\end{verbatim}
\item \textbf{Figure regeneration from local rerun outputs} (run from the public
  code-release \path{figures/} directory against that separate rerun tree; the
  quantitative figures use \texttt{OUTDIR}.  Regenerating the qualitative
  gallery additionally requires locally retained density arrays, which are not
  included in the arXiv source package):
\begin{verbatim}
python plot_figures.py --results-dir $OUTDIR --figs-dir $FIGDIR
python make_3d_renders.py --renders-dir $GALLERYDIR --figs-dir $FIGDIR
\end{verbatim}
\end{enumerate}

\begingroup
\small
\setlength{\LTleft}{0pt}
\setlength{\LTright}{0pt}
\begin{longtable}{>{\RaggedRight\arraybackslash}p{0.27\linewidth}>{\RaggedRight\arraybackslash}p{0.20\linewidth}>{\RaggedRight\arraybackslash}p{0.45\linewidth}}
\caption{Non-trivial run defaults and fixed seeds used by the reported runs.}
\label{tab:repro_defaults}\\
\toprule
Setting & Value & Scope / note\\
\midrule
\endfirsthead
\caption[]{Non-trivial run defaults and fixed seeds used by the reported runs (continued).}\\
\toprule
Setting & Value & Scope / note\\
\midrule
\endhead
Default hierarchy depth & 4 levels & Benchmark-suite runs unless an experiment states otherwise; the E6 sensitivity sweep also screens 3-level runs, and M8 uses 3 levels.\\
Fine smoother degree & 2 & Default for the paper runs; E6(d) additionally tests degree 4.\\
Coarse pre/post smoothing & 2 steps & Implemented as \texttt{coarse\_smooth\_iters=2} in the hierarchy.\\
Jacobi damping cap & $\omega_{\max}=0.5$ & Used by the Jacobi fallback and coarse-level smoothers.\\
Chebyshev lower fraction & $\alpha=1/30$ & Default spectral target interval for all Chebyshev runs.\\
Dense-Cholesky cutoff & 5000 coarse DOFs & Above this threshold the coarsest solve switches to fixed-count PCG.\\
Iterative coarsest solve budget & 80 PCG steps & Used only when the coarsest dense Cholesky path is unavailable.\\
Default outer restart & 32 & Default FGMRES restart in the solver-comparison paths; the E6 sensitivity surface screens 16/32/50, and the E7 large-scale plus E10 robustness runs use restart 50.\\
$\kappa_{\mathrm{eff}}$ start vector & normalized Gaussian draw with CuPy seed 0 & Deterministic Lanczos initial vector for the empirical $\kappa_{\mathrm{eff}}$ probe in the spectral-proxy, robustness, and small diagnostic utilities.\\
Heterogeneous sweep seed & 42 & Shared Bernoulli voxel-placement seed for the binary-contrast spectral-proxy probes.\\
E10 robustness-case seeds & 7, 11, 13, 17, 19 & Binary cases use 7/11/13; the near-singular random case with $\rho_{\mathrm{floor,test}}=10^{-12}$ uses 17; the mixed near-void case uses 19; checkerboard and layered-band are deterministic constructions.\\
E10 basin seed & 23 & Fixed seed for the broader $(V_f,\rho_{\mathrm{floor,test}},p)$ basin sweep.\\
\bottomrule
\end{longtable}
\endgroup

Hardware requirements: NVIDIA GPU with SM~8.0+ (Ampere/Ada) and
$\ge$12\,GiB VRAM for the smaller experiments.  The reported 1M-element run in
\S\ref{ssec:exp_e7} was measured on a 24\,GiB RTX~4090; the setup-time
hierarchy-allocation delta is 8.66\,GiB in the E7 large-scale output.
Expected runtimes: pre-benchmark validation checks $\approx$30\,min; the main
benchmark suite $\approx$4--6\,h at the default statistical setting
(two warm-ups, ten timed trials for repeated-run measurements).
The full benchmark suite was not energy-metered.  At the stated runtime on one
RTX~4090 workstation, total benchmark energy should be interpreted as an
unreported setup and measurement cost rather than inferred from the E9
solve-window proxy; a rough workstation-level order of magnitude is a few kWh,
depending on host-system power, idle intervals, failed or capped screening
runs, and GPU frequency behavior.

\section{Code and Provenance Map}\label{app:artifact}

Table~\ref{tab:artifact_modules} lists reader-facing implementation components;
Table~\ref{tab:artifact_results} maps each numbered result in the manuscript to
its reproduction step and output record type.  Exact release paths, raw module
filenames, and stored gallery-array names are maintained in the public
\path{ARTIFACT_MAP.md} rather than repeated in full in the manuscript.  The
following tables are reproduction metadata, not additional scientific claims.

\begin{table}[H]
\centering
\small
\caption{Reader-facing implementation-component map.}
\label{tab:artifact_modules}
\begin{tabularx}{0.95\linewidth}{>{\RaggedRight\arraybackslash}p{0.28\linewidth}>{\RaggedRight\arraybackslash}p{0.27\linewidth}Y}
\toprule
Reader-facing component & Released implementation component & Purpose\\
\midrule
GMG hierarchy implementation & Hierarchy construction module & Mixed-precision hierarchy construction and Lanczos-based empirical $\kappa_{\mathrm{eff}}$ estimate\\
SIMP solver integration & GMG-enabled SIMP driver & Matrix-free SIMP driver with the GMG solver path\\
Baseline Jacobi / matrix-free reference path & Baseline solver implementation & Jacobi-PCG and matrix-free baseline used in E2, E3, and E8\\
Level-0 fused matvec & Fused Level-0 operator implementation & BF16 WMMA gather--GEMM--scatter kernel; companion paper \citep{yang2026fused}\\
Validation-check workflow & Validation driver & M1--M8 gates and retained validation record\\
Small-problem spectrum / symmetry diagnostic & Spectrum and symmetry diagnostic & Probe used only to interpret the PCG/FGMRES policy and empirical $\kappa_{\mathrm{eff}}$ screen\\
Benchmark-suite workflow & E1--E10 benchmark driver & E1--E10 plus direct BF16 and high-contrast smoother diagnostics\\
Figure-generation workflow & Quantitative and qualitative figure workflows & Regenerates quantitative figures and the qualitative gallery\\
Conda environment specification & Reproducibility environment file & Recreate-from-scratch package pins\\
\bottomrule
\end{tabularx}
\end{table}

\begingroup
\small
\setlength{\LTleft}{0pt}
\setlength{\LTright}{0pt}
\begin{longtable}{>{\RaggedRight\arraybackslash}p{0.31\linewidth}>{\RaggedRight\arraybackslash}p{0.29\linewidth}>{\RaggedRight\arraybackslash}p{0.32\linewidth}}
\caption{Mapping from reported results to reproduction commands and output artifacts.}
\label{tab:artifact_results}\\
\toprule
Manuscript result & Reproduction step & Output artifact type\\
\midrule
\endfirsthead
\caption[]{Mapping from reported results to reproduction commands and output artifacts (continued).}\\
\toprule
Manuscript result & Reproduction step & Output artifact type\\
\midrule
\endhead
\S\ref{sec:method}, Figs.~\ref{fig:method_overview}--\ref{fig:precision_overview}: conceptual method diagrams & Hand-crafted manuscript illustrations & Conceptual figure files, not generated by the benchmark plotting pipeline\\
\S\ref{ssec:phase1}, Table~\ref{tab:phase1}: validation checks & Validation-check driver, M1--M8 gates & Validation result JSON\\
\S\ref{ssec:exp_e1}, Fig.~\ref{fig:e1}: outer iter.\ count vs.\ mesh & Benchmark-suite driver, E1 heterogeneous sweep & E1 iteration-count summary\\
\S\ref{ssec:exp_e2}, Figs.~\ref{fig:e2a}--\ref{fig:e2b}, Table~\ref{tab:e2}: per-solve wall time & Benchmark-suite driver, E2 repeated solve timings & E2 timing summary and trial records\\
\S\ref{ssec:exp_e2}, Fig.~\ref{fig:e2c}: residual histories & Benchmark-suite driver, E2 residual logging & E2 residual-history record\\
Appendix~\ref{app:e3_aux}, Fig.~\ref{fig:e3}, Table~\ref{tab:e3}: auxiliary 30-step schedule timing & Benchmark-suite driver, E3 fixed 30-step schedule & E3 schedule-timing summary\\
Appendix~\ref{app:e3_aux}, Fig.~\ref{fig:e3b}: SIMP trajectory diagnostic & Benchmark-suite driver, E3 trajectory logging & E3 trajectory record\\
\S\ref{ssec:exp_e4}, Fig.~\ref{fig:e4}: tensor-core throughput & Benchmark-suite driver, E4 200-rep kernel proxy & E4 throughput summary\\
\S\ref{ssec:exp_e4}, Fig.~\ref{fig:e4b}: roofline proxy & Benchmark-suite driver, E4 roofline proxy & E4 roofline record\\
\S\ref{ssec:exp_e5}, Fig.~\ref{fig:e5}: $\kappa_{\mathrm{eff}}$ spectral proxy & Benchmark-suite driver, E5 heterogeneous spectral-proxy sweep & E5 spectral-proxy summary\\
\S\ref{ssec:exp_e5}, Table~\ref{tab:e5_direct}: direct BF16 validation & Benchmark-suite driver, heterogeneous BF16/FP64 checks & E5 direct-validation summary and residual histories\\
\S\ref{ssec:exp_e6}, Fig.~\ref{fig:e6}: ablations (a)--(d) & Benchmark-suite driver, E6 ablations & E6 ablation summaries and trial records\\
\S\ref{ssec:exp_e6}, Fig.~\ref{fig:e6b}: sensitivity surface & Benchmark-suite driver, E6 screening surface & E6 sensitivity-surface summary\\
\S\ref{ssec:exp_e6}, Table~\ref{tab:e6_high_contrast}: high-contrast smoother screen & Benchmark-suite driver, high-contrast smoother checks & E6 high-contrast smoother summary and residual histories\\
\S\ref{ssec:exp_e7}, Fig.~\ref{fig:e7}, Table~\ref{tab:e7}: 1M-element scaling & Benchmark-suite driver, E7 large-scale repeated solves & E7 large-scale timing summary and trial records\\
\S\ref{ssec:exp_e8}, Fig.~\ref{fig:e8}: PyAMG external baseline & Benchmark-suite driver, E8 post-assembly assembled reference & E8 external-baseline timing summary\\
\S\ref{ssec:exp_e9}: energy efficiency & Benchmark-suite driver, E9 NVML energy probe & E9 solve-window energy summary and trial records\\
\S\ref{ssec:exp_e10}, Fig.~\ref{fig:e10}, Table~\ref{tab:e10}: robustness cases & Benchmark-suite driver, E10 robustness cases & E10 robustness-case summary\\
\S\ref{ssec:exp_e10}, Fig.~\ref{fig:e10b}: collapsed single-seed screening map & Benchmark-suite driver, E10 basin sweep & E10 basin-screening summary\\
\bottomrule
\end{longtable}
\endgroup

\begin{table}[H]
\centering
\small
\caption{Mapping from reader-facing qualitative labels to stable gallery sources.}
\label{tab:gallery_sources}
\begin{tabularx}{\linewidth}{>{\RaggedRight\arraybackslash}p{0.24\linewidth}>{\RaggedRight\arraybackslash}p{0.15\linewidth}>{\RaggedRight\arraybackslash}p{0.30\linewidth}Y}
\toprule
Reader-facing label & Stable panel ID & Stored source category & Description\\
\midrule
Cantilever 216k & Panel C216k & selected auxiliary density field & best-valid auxiliary density field\\
Bridge 216k & Panel B216k & selected auxiliary density field & best-valid auxiliary density field\\
Double-clamped beam 216k & Panel D216k & selected auxiliary density field & best-valid auxiliary density field\\
Torsion 499k & Panel T499k & selected auxiliary density field & best-valid auxiliary density field stored under a rounded 500k gallery identifier\\
Cantilever 1M & Panel C1M & selected auxiliary density field & best-valid auxiliary density field\\
Cantilever 1M low-volume qualitative snapshot & Panel C1M-LV & retained low-volume qualitative snapshot & retained qualitative snapshot\\
\bottomrule
\end{tabularx}
\end{table}

\noindent
The low-volume 1M cantilever panel is a retained qualitative snapshot rather
than evidence for a distinct final-versus-best state.  Exact array and metadata
filenames are listed in the public \path{ARTIFACT_MAP.md}.  These retained
snapshots are gallery inputs only, not regenerated quantitative E1--E10 results.

\section[Idealised Conditional Argument for the kappa-eff Bound]{Idealised Conditional Argument for the $\kappa_{\mathrm{eff}}$ Bound}\label{app:bound}

This appendix gives the limited idealised argument behind the
$\eps_{\mathrm{BF16}}\cdot\kappa_{\mathrm{eff}}$ screen.  It is not a theorem
for the exact floating-point mixed-precision V-cycle used in the code, and it
does not prove BF16 convergence for every tested density state.  Its role is to
motivate a measurable diagnostic in the frozen linear/SPD setting.

\subsection{Setup and Notation}

Let $\bK$ denote the free-DOF fine-system operator with
eigenvalues $0<\lambda_1\le\cdots\le\lambda_n$.
Let $\mathcal{M}$ denote one application of the frozen V-cycle as a linear
operator in the idealised SPD interpretation used for the PCG pairings.
The effective condition number is
\begin{equation}
  \kappa_{\mathrm{eff}} = \frac{\lambda_{\max}(\mathcal{M}\bK)}{\lambda_{\min}(\mathcal{M}\bK)}.
\end{equation}
The experiments estimate this quantity with a Lanczos probe.  When the
implemented preconditioned operator is treated conservatively as non-symmetric,
the solver uses FGMRES and $\kappa_{\mathrm{eff}}$ is interpreted only as a
screening scalar, not as a complete convergence predictor.

\subsection{V-Cycle Error Propagation}

The error after one V-cycle satisfies $\be_{k+1} = \mathcal{E}\be_k$
where the error propagation operator is
\begin{equation}
  \mathcal{E} = (\bI - \mathcal{S}_{\mathrm{post}}\bK)(\bI - \bP(\bK^1)^{-1}\bR\bK)(\bI - \mathcal{S}_{\mathrm{pre}}\bK).
\end{equation}
Here $\mathcal{S}\bK = q_\nu(\bK)\bK$, where $q_\nu$ is the degree-$\nu$
smoothing polynomial.

\begin{proposition}[Chebyshev residual estimate]\label{prop:rho}
For a Chebyshev-Jacobi smoother of degree $\nu$ targeting
$[\alpha\lambda_{\max}, \lambda_{\max}]$, the residual polynomial on that
targeted interval satisfies
\begin{equation}
  \max_{\lambda\in[\alpha\lambda_{\max},\lambda_{\max}]}
  |p_\nu(\lambda)|
  \le \frac{1}{T_\nu(\sigma/\delta)},
  \qquad
  \sigma=\frac{(1+\alpha)\lambda_{\max}}{2},\quad
  \delta=\frac{(1-\alpha)\lambda_{\max}}{2}.
\end{equation}
For the paper default $\alpha=1/30$ and $\nu=2$,
$\sigma/\delta=31/29$ and $1/T_2(31/29)\approx0.78$.
\end{proposition}

\begin{proof}[Proof sketch]
This is the standard Chebyshev semi-iteration residual-polynomial bound on the
targeted spectral interval.  The numerical value follows from
$T_2(x)=2x^2-1$.  The result is a targeted-band residual estimate only; it is
not a complete bound on the multigrid error-propagation operator because the
coarse-space approximation, boundary mask, and elasticity near-null components
also enter $\mathcal{E}$.
\end{proof}

The practical V-cycle smoothing factor is therefore an empirical property of
the combined smoother and Galerkin coarse correction, not a direct consequence
of the scalar $1/T_\nu(\sigma/\delta)$ alone.

\subsection{Effective Condition Number}

\begin{proposition}[Idealised effective-condition-number estimate]\label{prop:kappa}
Assume the frozen V-cycle is an SPD preconditioner for $\bK$ and its
error-propagation operator satisfies $\|\mathcal{E}\|_{\bK}\le\rho<1$ in the
energy norm.  Then
\begin{equation}
  \kappa_{\mathrm{eff}} = \kappa(\mathcal{M}\bK) \le
  \frac{1+\rho}{1-\rho}.
  \label{eq:kappa_formula}
\end{equation}
For a non-normal or only approximately SPD floating-point implementation, this
bound is only an interpretive idealisation.
\end{proposition}

\begin{proof}[Proof sketch]
In the ideal SPD case, $\mathcal{E}=\bI-\mathcal{M}\bK$ is self-adjoint in the
energy inner product.  If $\|\mathcal{E}\|_{\bK}\le\rho$, then the eigenvalues
of $\mathcal{M}\bK$ lie in $[1-\rho,1+\rho]$
(standard multigrid convergence analysis, e.g.,
\citet{trottenberg2000multigrid}).
Taking the ratio of extremes gives~\eqref{eq:kappa_formula}.
\end{proof}

\subsection{Implication for BF16 Stability}

The mixed-precision literature motivates the screen
$\eps_{\mathrm{BF16}}\kappa_{\mathrm{eff}}<1$: if the preconditioned spectrum
is small enough, BF16 fine-level work is less likely to dominate the outer
iteration error.  In this paper the screen is evaluated empirically in
Section~\ref{ssec:exp_e5}.  It is satisfied on 17 of 18 fixed-seed
heterogeneous probe cases and violated on one near-void 216k case.  Direct
BF16 validation on those same states shows 11 screened-in capped runs and one
screened-out converged run, so the screen should be read as an interpretive
diagnostic rather than as an admissibility certificate for the implemented
floating-point hierarchy.

\begin{remark}
For degree-2 Chebyshev smoothing, the targeted-band value
$1/T_2(31/29)\approx0.78$ is modest.  The smoother can still be useful inside a
Galerkin hierarchy, but the reported ablations support that claim through the
reported ablations and convergence screens rather than through a standalone
$0.10$ polynomial damping guarantee.
\end{remark}

\bibliographystyle{unsrtnat}
\bibliography{refs_references}

@article{amir2014multigrid,
  title     = {On multigrid-{CG} for efficient topology optimization},
  author    = {Amir, Oded and Aage, Niels and Lazarov, Boyan S.},
  journal   = {Structural and Multidisciplinary Optimization},
  volume    = {49},
  number    = {5},
  pages     = {815--829},
  year      = {2014},
  doi       = {10.1007/s00158-013-1015-5}
}

@article{wu2016system,
  title     = {A system for high-resolution topology optimization},
  author    = {Wu, Jun and Dick, Christian and Westermann, R\"{u}diger},
  journal   = {IEEE Transactions on Visualization and Computer Graphics},
  volume    = {22},
  number    = {3},
  pages     = {1195--1208},
  year      = {2016},
  doi       = {10.1109/TVCG.2015.2502588}
}

@article{wang2025matrixfree,
  title     = {Efficient large-scale 3D topology optimization with matrix-free MATLAB code},
  author    = {Wang, Junpeng and Aage, Niels and Wu, Jun and Sigmund, Ole and Westermann, R{\"u}diger},
  journal   = {Structural and Multidisciplinary Optimization},
  volume    = {68},
  number    = {9},
  pages     = {174},
  year      = {2025},
  doi       = {10.1007/s00158-025-04127-3}
}

@article{lazarov2011filters,
  title     = {Filters in topology optimization based on {Helmholtz}-type differential equations},
  author    = {Lazarov, Boyan S. and Sigmund, Ole},
  journal   = {International Journal for Numerical Methods in Engineering},
  volume    = {86},
  number    = {6},
  pages     = {765--781},
  year      = {2011},
  doi       = {10.1002/nme.3072}
}

@article{carson2017new,
  title     = {A new analysis of iterative refinement and its application to accurate solution
               of ill-conditioned sparse linear systems},
  author    = {Carson, Erin and Higham, Nicholas J.},
  journal   = {SIAM Journal on Scientific Computing},
  volume    = {39},
  number    = {6},
  pages     = {A2834--A2856},
  year      = {2017},
  doi       = {10.1137/17M1122918}
}

@article{higham2022mixed,
  title     = {Mixed precision algorithms in numerical linear algebra},
  author    = {Higham, Nicholas J. and Mary, Theo},
  journal   = {Acta Numerica},
  volume    = {31},
  pages     = {347--414},
  year      = {2022},
  doi       = {10.1017/S0962492922000022}
}

@article{haidar2020tensor,
  title     = {Mixed-precision iterative refinement using tensor cores on {GPUs} to accelerate solution of linear systems},
  author    = {Haidar, Azzam and Bayraktar, Harun and Tomov, Stanimire and Dongarra, Jack and Higham, Nicholas J.},
  journal   = {Proceedings of the Royal Society A: Mathematical, Physical and Engineering Sciences},
  volume    = {476},
  number    = {2243},
  pages     = {20200110},
  year      = {2020},
  doi       = {10.1098/rspa.2020.0110}
}

@article{buttari2007dense,
  title     = {Mixed precision iterative refinement techniques for the solution of dense linear systems},
  author    = {Buttari, Alfredo and Dongarra, Jack and Langou, Julie and Langou, Julien and Luszczek, Piotr and Kurzak, Jakub},
  journal   = {The International Journal of High Performance Computing Applications},
  volume    = {21},
  number    = {4},
  pages     = {457--466},
  year      = {2007},
  doi       = {10.1177/1094342007084026}
}

@article{buttari2008sparse,
  title     = {Using mixed precision for sparse matrix computations to enhance the performance while achieving 64-bit accuracy},
  author    = {Buttari, Alfredo and Dongarra, Jack and Kurzak, Jakub and Luszczek, Piotr and Tomov, Stanimir},
    journal   = {ACM Transactions on Mathematical Software},
    volume    = {34},
    number    = {4},
    pages     = {17:1--17:22},
    year      = {2008},
    doi       = {10.1145/1377596.1377597}
  }

@article{goddeke2007mixedfem,
  title     = {Performance and accuracy of hardware-oriented native-, emulated- and mixed-precision solvers in {FEM} simulations},
  author    = {G{\"o}ddeke, Dominik and Strzodka, Robert and Turek, Stefan},
  journal   = {International Journal of Parallel, Emergent and Distributed Systems},
  volume    = {22},
  number    = {4},
  pages     = {221--256},
  year      = {2007},
  doi       = {10.1080/17445760601122076}
}

@article{goddeke2011cyclic,
  title     = {Cyclic reduction tridiagonal solvers on {GPUs} applied to mixed precision multigrid},
  author    = {G{\"o}ddeke, Dominik and Strzodka, Robert},
  journal   = {IEEE Transactions on Parallel and Distributed Systems},
  volume    = {22},
  number    = {1},
  pages     = {22--32},
  year      = {2011},
  doi       = {10.1109/TPDS.2010.61}
}

@inproceedings{abdelfattah2019batched,
  title     = {Fast batched matrix multiplication for small sizes using half-precision arithmetic on {GPUs}},
  author    = {Abdelfattah, Ahmad and Tomov, Stanimire and Dongarra, Jack},
  booktitle = {Proceedings - 2019 IEEE 33rd International Parallel and Distributed Processing Symposium, IPDPS 2019},
  pages     = {111--122},
  year      = {2019},
  doi       = {10.1109/IPDPS.2019.00022}
}

@article{abdelfattah2020batched,
  title     = {Matrix multiplication on batches of small matrices in half and half-complex precisions},
  author    = {Abdelfattah, Ahmad and Tomov, Stanimire and Dongarra, Jack},
  journal   = {Journal of Parallel and Distributed Computing},
  volume    = {145},
  pages     = {188--201},
  year      = {2020},
  doi       = {10.1016/j.jpdc.2020.07.001}
}

@article{lopez2023mixedlu,
  title     = {Mixed precision {LU} factorization on {GPU} tensor cores: reducing data movement and memory footprint},
  author    = {Lopez, Florent and Mary, Theo},
  journal   = {The International Journal of High Performance Computing Applications},
  volume    = {37},
  number    = {2},
  pages     = {165--179},
  year      = {2023},
  doi       = {10.1177/10943420221136848}
}

@article{naumov2015amgx,
  title     = {{AmgX}: A Library for {GPU} Accelerated Algebraic Multigrid and Preconditioned
               Iterative Methods},
  author    = {Naumov, Maxim and Arsaev, Marat and Castonguay, Patrice and Cohen, Jonathan
               and Demouth, Julien and Eaton, Joe and Layton, Simon and Markovskiy, Nikolay
               and Reguly, Istvan and Sakharnykh, Nikolai and Sellappan, Vijay and Strzodka, Robert},
  journal   = {SIAM Journal on Scientific Computing},
  volume    = {37},
  number    = {5},
  pages     = {S602--S626},
  year      = {2015},
  doi       = {10.1137/140980260}
}

@book{trottenberg2000multigrid,
  title     = {Multigrid},
  author    = {Trottenberg, Ulrich and Oosterlee, Cornelius W. and Sch\"{u}ller, Anton},
  publisher = {Academic Press},
  year      = {2000},
  address   = {San Diego, CA}
}

@book{briggs2000multigrid,
  title     = {A Multigrid Tutorial, Second Edition},
  author    = {Briggs, William L. and Henson, Van Emden and McCormick, Steve F.},
  edition   = {2nd},
  publisher = {Society for Industrial and Applied Mathematics},
  year      = {2000},
  address   = {Philadelphia, PA},
  doi       = {10.1137/1.9780898719505}
}

@inproceedings{lorensen1987marching,
  title     = {Marching Cubes: A High Resolution {3D} Surface Construction Algorithm},
  author    = {Lorensen, William E. and Cline, Harvey E.},
  booktitle = {Proceedings of the 14th Annual Conference on Computer Graphics and Interactive Techniques},
  series    = {SIGGRAPH '87},
  pages     = {163--169},
  year      = {1987},
  publisher = {ACM},
  doi       = {10.1145/37401.37422}
}

@book{saad2003iterative,
  title     = {Iterative Methods for Sparse Linear Systems},
  author    = {Saad, Yousef},
  edition   = {2nd},
  publisher = {Society for Industrial and Applied Mathematics},
  year      = {2003},
  address   = {Philadelphia, PA},
  doi       = {10.1137/1.9780898718003}
}

@article{bendsoe1988homogenization,
  title     = {Generating optimal topologies in structural design using a homogenization method},
  author    = {Bends{\o}e, Martin Philip and Kikuchi, Noboru},
  journal   = {Computer Methods in Applied Mechanics and Engineering},
  volume    = {71},
  number    = {2},
  pages     = {197--224},
  year      = {1988},
  doi       = {10.1016/0045-7825(88)90086-2}
}

@article{bourdin2001filters,
  title     = {Filters in topology optimization},
  author    = {Bourdin, Blaise},
  journal   = {International Journal for Numerical Methods in Engineering},
  volume    = {50},
  number    = {9},
  pages     = {2143--2158},
  year      = {2001},
  doi       = {10.1002/nme.116}
}

@article{andreassen201188,
  title     = {Efficient topology optimization in {MATLAB} using 88 lines of code},
  author    = {Andreassen, Erik and Clausen, Anders and Schevenels, Mattias and Lazarov, Boyan S.
               and Sigmund, Ole},
  journal   = {Structural and Multidisciplinary Optimization},
  volume    = {43},
  number    = {1},
  pages     = {1--16},
  year      = {2011},
  doi       = {10.1007/s00158-010-0594-7}
}

@article{guest2004minimum,
  title     = {Achieving minimum length scale in topology optimization using nodal design variables and projection functions},
  author    = {Guest, James K. and Pr{\'e}vost, Jean H. and Belytschko, Ted},
  journal   = {International Journal for Numerical Methods in Engineering},
  volume    = {61},
  number    = {2},
  pages     = {238--254},
  year      = {2004},
  doi       = {10.1002/nme.1064}
}

@article{sigmund2013review,
  title     = {Topology optimization approaches},
  author    = {Sigmund, Ole and Maute, Kurt},
  journal   = {Structural and Multidisciplinary Optimization},
  volume    = {48},
  number    = {6},
  pages     = {1031--1055},
  year      = {2013},
  doi       = {10.1007/s00158-013-0978-6}
}

@article{wang2003levelset,
  title     = {A level set method for structural topology optimization},
  author    = {Wang, Michael Yu and Wang, Xiaoming and Guo, Dongming},
  journal   = {Computer Methods in Applied Mechanics and Engineering},
  volume    = {192},
  number    = {1--2},
  pages     = {227--246},
  year      = {2003},
  doi       = {10.1016/S0045-7825(02)00559-5}
}

@article{allaire2004levelset,
  title     = {Structural optimization using sensitivity analysis and a level-set method},
  author    = {Allaire, Gr{\'e}goire and Jouve, Fran{\c c}ois and Toader, Anca-Maria},
  journal   = {Journal of Computational Physics},
  volume    = {194},
  number    = {1},
  pages     = {363--393},
  year      = {2004},
  doi       = {10.1016/j.jcp.2003.09.032}
}

@article{bendsoe1999material,
  title     = {Material interpolation schemes in topology optimization},
  author    = {Bends{\o}e, Martin P. and Sigmund, Ole},
  journal   = {Archive of Applied Mechanics (Ingenieur Archiv)},
  volume    = {69},
  number    = {9--10},
  pages     = {635--654},
  year      = {1999},
  doi       = {10.1007/s004190050248}
}

@article{traff2023simple,
  title     = {Simple and efficient {GPU} accelerated topology optimisation: Codes and applications},
  author    = {Tr{\"a}ff, Erik A. and Rydahl, Anton and Karlsson, Sven and Sigmund, Ole and Aage, Niels},
  journal   = {Computer Methods in Applied Mechanics and Engineering},
  volume    = {410},
  pages     = {116043},
  year      = {2023},
  doi       = {10.1016/j.cma.2023.116043}
}

@article{ferrari2020top99neo,
  title     = {A new generation 99 line {Matlab} code for compliance topology optimization and its extension to 3D},
  author    = {Ferrari, Federico and Sigmund, Ole},
  journal   = {Structural and Multidisciplinary Optimization},
  volume    = {62},
  number    = {4},
  pages     = {2211--2228},
  year      = {2020},
  doi       = {10.1007/s00158-020-02629-w}
}

@article{yago2022comparative,
  title     = {Topology optimization methods for 3D structural problems: A comparative study},
  author    = {Yago, Daniel and Cante, Juan and Lloberas-Valls, Oriol and Oliver, Javier},
  journal   = {Archives of Computational Methods in Engineering},
  volume    = {29},
  number    = {3},
  pages     = {1525--1567},
  year      = {2022},
  doi       = {10.1007/s11831-021-09626-2}
}

@article{aage2013parallel,
  title     = {Parallel framework for topology optimization using the method of moving asymptotes},
  author    = {Aage, Niels and Lazarov, Boyan S.},
  journal   = {Structural and Multidisciplinary Optimization},
  volume    = {47},
  number    = {4},
  pages     = {493--505},
  year      = {2013},
  doi       = {10.1007/s00158-012-0869-2}
}

@article{aage2015petsc,
  title     = {Topology optimization using {PETSc}: An easy-to-use, fully parallel, open source topology optimization framework},
  author    = {Aage, Niels and Andreassen, Erik and Lazarov, Boyan S.},
  journal   = {Structural and Multidisciplinary Optimization},
  volume    = {51},
  number    = {3},
  pages     = {565--572},
  year      = {2015},
  doi       = {10.1007/s00158-014-1157-0}
}

@article{aage2017giga,
  title     = {Giga-voxel computational morphogenesis for structural design},
  author    = {Aage, Niels and Andreassen, Erik and Lazarov, Boyan S. and Sigmund, Ole},
  journal   = {Nature},
  volume    = {550},
  number    = {7674},
  pages     = {84--86},
  year      = {2017},
  doi       = {10.1038/nature23911}
}

@article{herrero2023adaptive,
  title     = {A parallel geometric multigrid method for adaptive topology optimization},
  author    = {Herrero-P{\'e}rez, David and Pic{\'o}-Vicente, Sebasti{\'a}n Gin{\'e}s},
  journal   = {Structural and Multidisciplinary Optimization},
  volume    = {66},
  number    = {10},
  pages     = {225},
  year      = {2023},
  doi       = {10.1007/s00158-023-03675-w}
}

@misc{nvidia2023ada,
  title     = {{NVIDIA} {Ada} {GPU} Architecture},
  author    = {{NVIDIA}},
  year      = {2023},
  note      = {NVIDIA white paper},
  url       = {https://images.nvidia.com/aem-dam/Solutions/geforce/ada/nvidia-ada-gpu-architecture.pdf}
}

@inproceedings{markidis2018nvidia,
  title     = {{NVIDIA} {Tensor} {Core} Programmability, Performance \& Precision},
  author    = {Markidis, Stefano and Chien, Steven Wei Der and Laure, Erwin and Peng, Ivy Bo and
               Vetter, Jeffrey S.},
  booktitle = {2018 IEEE International Parallel and Distributed Processing Symposium Workshops (IPDPSW)},
  pages     = {522--531},
  year      = {2018},
  doi       = {10.1109/IPDPSW.2018.00091}
}

@inproceedings{tsai2023mixed,
  title     = {Mixed Precision Algebraic Multigrid on {GPUs}},
  author    = {Tsai, Yu-Hsiang Mike and Beams, Natalie and Anzt, Hartwig},
  booktitle = {Parallel Processing and Applied Mathematics (PPAM 2022)},
  series    = {Lecture Notes in Computer Science},
  volume    = {13826},
  pages     = {113--125},
  year      = {2023},
  doi       = {10.1007/978-3-031-30442-2_9}
}

@article{tsai2023three,
  title     = {Three-precision algebraic multigrid on {GPUs}},
  author    = {Tsai, Yu-Hsiang Mike and Beams, Natalie and Anzt, Hartwig},
  journal   = {Future Generation Computer Systems},
  volume    = {149},
  pages     = {280--293},
  year      = {2023},
  doi       = {10.1016/j.future.2023.07.024}
}

@misc{bohm2025galerkin,
    title     = {Large-scale Multigrid with Adaptive {Galerkin} Coarsening},
    author    = {B{\"o}hm, Fabian and Kohl, Nils and K{\"o}stler, Harald and R{\"u}de, Ulrich},
    howpublished = {arXiv preprint arXiv:2511.13109},
    year      = {2025},
    doi       = {10.48550/arXiv.2511.13109},
    url       = {https://arxiv.org/abs/2511.13109}
  }

@article{padhi2023gpu,
  title     = {Efficient hybrid topology optimization using {GPU} and homogenization-based
               multigrid approach},
  author    = {Padhi, Arya Prakash and Chakraborty, Souvik and Chakrabarti, Anupam
               and Chowdhury, Rajib},
  journal   = {Engineering with Computers},
  volume    = {39},
  number    = {5},
  pages     = {3593--3615},
  year      = {2023},
  doi       = {10.1007/s00366-022-01771-x}
}

@article{zhao2024gpu,
  title     = {Efficient {GPU} accelerated topology optimization of composite structures
               with spatially varying fiber orientations},
  author    = {Zhao, Junpeng and Qi, Tianyuan and Wang, Chunjie},
  journal   = {Computer Methods in Applied Mechanics and Engineering},
  volume    = {421},
  pages     = {116809},
  year      = {2024},
  doi       = {10.1016/j.cma.2024.116809}
}

@article{anderson2021mfem,
  title     = {{MFEM}: A Modular Finite Element Methods Library},
  author    = {Anderson, Robert and Andrej, Julian and Barker, Andrew and Bramwell, Jamie and Camier, Jean-Sylvain and Cerveny, Jakub and Dobrev, Veselin and Dudouit, Yohann and Fisher, Aaron and Kolev, Tzanio and Pazner, Will and Stowell, Mark and Tomov, Vladimir and Akkerman, Ido and Dahm, Johann and Medina, David and Zampini, Stefano},
  journal   = {Computers \& Mathematics with Applications},
  volume    = {81},
  pages     = {42--74},
  year      = {2021},
  doi       = {10.1016/j.camwa.2020.06.009}
}

@article{davydov2020matrixfree,
  title     = {A matrix-free approach for finite-strain hyperelastic problems using geometric multigrid},
  author    = {Davydov, Denis and Pelteret, Jean-Paul and Arndt, Daniel and Kronbichler, Martin and Steinmann, Paul},
  journal   = {International Journal for Numerical Methods in Engineering},
  volume    = {121},
  number    = {13},
  pages     = {2874--2895},
  year      = {2020},
  doi       = {10.1002/nme.6336}
}

@article{franco2020lor,
  title     = {High-order matrix-free incompressible flow solvers with {GPU} acceleration and low-order refined preconditioners},
  author    = {Franco, Michael and Camier, Jean-Sylvain and Andrej, Julian and Pazner, Will},
  journal   = {Computers \& Fluids},
  volume    = {203},
  pages     = {104541},
  year      = {2020},
  doi       = {10.1016/j.compfluid.2020.104541}
}

@article{vargas2022matrixfree,
  title     = {Matrix-free approaches for {GPU} acceleration of a high-order finite element hydrodynamics application using {MFEM}, {Umpire}, and {RAJA}},
  author    = {Vargas, Arturo and Stitt, Thomas M. and Weiss, Kenneth and Tomov, Vladimir Z. and Camier, Jean-Sylvain and Kolev, Tzanio and Rieben, Robert N.},
  journal   = {The International Journal of High Performance Computing Applications},
  volume    = {36},
  number    = {4},
  pages     = {492--509},
  year      = {2022},
  doi       = {10.1177/10943420221100262}
}

@article{sun2020vectorization,
  title     = {A study of vectorization for matrix-free finite element methods},
  author    = {Sun, Tianjiao and Mitchell, Lawrence and Kulkarni, Kaushik and Kl{\"o}ckner, Andreas and Ham, David A. and Kelly, Paul H. J.},
  journal   = {The International Journal of High Performance Computing Applications},
  volume    = {34},
  number    = {6},
  pages     = {629--644},
  year      = {2020},
  doi       = {10.1177/1094342020945005}
}

@incollection{gavranovic2019gpgpu,
  title     = {Topology Optimization Using {GPGPU}},
  author    = {Gavranovic, Stefan and Hartmann, Dirk and Wever, Utz},
  booktitle = {Advances in Evolutionary and Deterministic Methods for Design, Optimization and Control in Engineering and Sciences},
  pages     = {553--566},
  year      = {2019},
  publisher = {Springer International Publishing},
  address   = {Cham},
  doi       = {10.1007/978-3-319-89988-6_33}
}

@article{kashi2026mixed,
    title     = {Mixed-precision numerics in scientific applications: survey and perspectives},
    author    = {Kashi, Aditya and Lu, Hao and Brewer, Wesley and Rogers, David and
                 Matheson, Michael and Shankar, Mallikarjun and Wang, Feiyi},
    journal   = {The Journal of Supercomputing},
    volume    = {82},
    number    = {5},
    pages     = {287},
    year      = {2026},
    doi       = {10.1007/s11227-026-08264-4},
    publisher = {Springer}
  }

@article{bell2023pyamg,
  title     = {{PyAMG}: Algebraic Multigrid Solvers in {Python}},
  author    = {Bell, Nathan and Olson, Luke N. and Schroder, Jacob and Southworth, Ben},
  journal   = {Journal of Open Source Software},
  volume    = {8},
  number    = {87},
  pages     = {5495},
  year      = {2023},
  doi       = {10.21105/joss.05495},
  publisher = {The Open Journal}
}

@misc{yang2026fused,
    title     = {Matrix-Free {3D} {SIMP} Topology Optimization with Fused Gather-{GEMM}-Scatter Kernels},
    author    = {Yang, Shaoliang and Wang, Jun and Wang, Yunsheng},
    howpublished = {arXiv preprint arXiv:2604.18020},
    year      = {2026},
    doi       = {10.48550/arXiv.2604.18020},
    url       = {https://arxiv.org/abs/2604.18020}
  }

\end{document}